\relax
\documentclass[letterpaper]{article} 
\usepackage{helvet} 
\usepackage{courier}  
\usepackage[hyphens]{url}  
\usepackage{graphicx} 
\usepackage[titlenumbered,ruled,vlined]{algorithm2e} 
\usepackage[labelsep=period]{caption} 
\usepackage[switch]{lineno}
\usepackage{rotating}
\usepackage[hidelinks]{hyperref}
\newif\iflong\longtrue

\usepackage{geometry}
\usepackage[english]{babel}
\usepackage[utf8]{inputenc}
\usepackage{amsmath,amssymb}
\usepackage{amsthm}
\usepackage[disable]{todonotes}
\usepackage{graphicx}
\usepackage{dsfont}
\usepackage{braket}
\usepackage[shortlabels]{enumitem}
\usepackage{tikz}
\usepackage{verbatim}
\usepackage{todonotes}
\usepackage{xspace}
\usepackage{cleveref}
\usepackage{nicefrac}
\usepackage{pdflscape} 

\newcommand\IV{$\Insert$\textsc{-Local \WBNSL}\xspace}
\newcommand\SW{$\Swap$\textsc{-Local \WBNSL}\xspace}

\newcommand{\W}[1]{\ensuremath{\mathrm{W}[#1]}\xspace}
\newcommand\NP{\ensuremath{\mathrm{NP}}\xspace}
\newcommand\FPT{\ensuremath{\mathrm{FPT}}\xspace}
\newcommand\XP{\ensuremath{\mathrm{XP}}\xspace}

\DeclareMathOperator{\Swap}{Swap}
\DeclareMathOperator{\Insert}{Insert}
\DeclareMathOperator{\Inv}{Inv}
\DeclareMathOperator{\improvable}{improvable}
\DeclareMathOperator{\Win}{Win}
\DeclareMathOperator{\WI}{WI}
\DeclareMathOperator{\CurrentScore}{currentScore}

\newcommand{\todog}{\todo[color=yellow!90!red]}

\usepackage[font=small,width=\textwidth]{caption}

\newcommand{\prob}[3]{\begin{center}
	\begin{minipage}[c]{.9\linewidth}
          \textsc{#1}\\
          \textbf{Input}: #2\\
          \textbf{Question}: #3
	\end{minipage}
\end{center}}

\newtheorem{theorem}{Theorem}
\newtheorem{lemma}[theorem]{Lemma}
\newtheorem{definition}[theorem]{Definition}
\newtheorem{proposition}[theorem]{Proposition}

\newtheorem{claim}{Claim}
\theoremstyle{definition}

{\bfseries}{\itshape}
{\bfseries}{\itshape}
\theoremstyle{definition}
\newtheorem*{claimproof}{\normalfont{\textit{Proof}}}

\definecolor{myred}{rgb}{1,0.25,0.25}

\newcommand{\Oh}{\mathcal{O}} 
\newcommand{\Mo}{\mathcal{M}} 
\newcommand{\Fa}{\mathcal{F}} 
\newcommand{\Pa}{\mathcal{P}} 

\newcommand{\w}{\omega}
\newcommand{\Con}{\mathcal{C}}
\newcommand{\score}{\text{\normalfont{score}}}

\newcommand{\poly}{\text{\normalfont{poly}}}

\newcommand{\WBNSL}{\textsc{W-BNSL}\xspace}
\newcommand{\MWBNSL}{\textsc{$\Inv$-Local \WBNSL}\xspace}
\newcommand{\wMWBNSL}{\textsc{${\rm InvWin}$-Local W-BNSL}\xspace}

\begin{document}

\title{Efficient Bayesian Network Structure Learning via\\ Parameterized Local Search on Topological Orderings}
\date{ }
\author {
    Niels Grüttemeier,
    Christian Komusiewicz {\normalfont and}
    Nils Morawietz\thanks{Supported by the Deutsche Forschungsgemeinschaft (DFG), project OPERAH, KO~3669/5-1.}  \\
}

\maketitle
\begin{abstract} In Bayesian Network Structure Learning (BNSL), one is given a variable
  set and parent scores for each variable and aims to compute a DAG, called Bayesian
  network, that maximizes the sum of parent scores, possibly under some structural
  constraints. Even very restricted special cases of BNSL are computationally hard, and, thus,
  in practice heuristics such as local search are used. A natural approach for a local search algorithm is a hill climbing strategy, where one replaces a given BNSL solution by a better solution within some
  pre-defined neighborhood as long as this is possible. We study ordering-based local search, where a
  solution is described via a topological ordering of the variables. We show that given
  such a topological ordering, one can compute an optimal DAG whose ordering is within
  inversion distance~$r$ in subexponential FPT time; the parameter~$r$ allows to balance between solution quality and running
  time of the local search algorithm. This running time bound can be achieved for BNSL without structural constraints and for all structural constraints that can be expressed via a sum
  of weights that are associated with each parent set. We also introduce a related distance called \emph{window inversions distance} and show that the corresponding local search problem can also be solved in subexponential FPT time for the parameter~$r$.  For two further natural modification operations on the variable orderings, we show that algorithms with an FPT time for~$r$ are unlikely.
  We also outline the limits of ordering-based local search by showing that it cannot be used for common
  structural constraints on the moralized graph of the network.
\end{abstract}

\section{Introduction}
Bayesian networks are arguably the most popular and important model for representing dependencies between variables of multivariate probability distributions~\cite{Dar10}. A Bayesian network consists of a directed acyclic graph (DAG)~$D$ and a set of conditional probability tables, one for each vertex of~$D$. The vertices of the network are the variables of the distribution and each conditional probability table specifies the probability distribution of the corresponding vertex given the values of its parents. Thus, the value of each variable depends directly on the values of its parents. The graph~$D$ is called the structure of the network. Given a set of observed data over some variable set~$N$, one needs to learn the structure~$D$ from this data. This is usually done in two steps. First, for each variable~$v$ and each possible set of parents a \emph{parent score}~$f_v(P)$ is computed. Roughly speaking, the score represents how useful it is to choose this particular set of parents in order to accurately represent the probability distribution that created the observed data.   Second, using the computed parent scores, we aim to compute the DAG that maximizes the sum of the parent scores of its vertices. This problem, called \textsc{Bayesian Network Structure Learning} (BNSL), is NP-hard~\cite{C95}.

The current-best worst-case running time bound for exact algorithms for BNSL is~$2^n\cdot \poly(|I|)$ where~$n$ is the number of variables and~$|I|$ is the size of the instance. Clearly, this running time is impractical for larger variable sets. Moreover, even very restricted special cases of BNSL remain NP-hard, for example the case where both every possible parent set has constant size \iflong and every variable has only a constant number of parent sets\fi{}~\cite{OS13}. Finally, restricting the topology of the DAG to sparse classes such as graphs of bounded treewidth or bounded degree leads to NP-hard learning problems as well~\cite{KP13,KP15,GK20}. 

Due to this notorious hardness of BNSL, it is mostly solved using heuristics. One of the most
successful heuristic approaches relies on local search~\cite{TBA06}. More precisely, 
in this approach one starts with an
arcless DAG~$D$ and adds, removes, or reverses an arc while this results in an increased
network score. 
More recent local search approaches do not use DAGs as solution representations but rather orderings of the variables~\cite{AOP11,LB17,SCZ17}. This approach is motivated by the fact that given an ordering~$\tau$ of the variables, one may greedily find the optimal DAG~$D$ among all DAGs for which~$\tau$ is a topological ordering~\cite{OS13}. Lee and van Beek~\cite{LB17} study a local search approach that is based on an operation, where one interchanges the positions of two consecutive vertices on the ordering as long as this results in an improvement of the network score. We refer to this operation as an~\emph{inversion}. Alonso-Barba et al.~\cite{AOP11} an \emph{insert} operation, where one removes a vertex from the ordering and inserts this vertex at a new position. Scanagatta et al.~\cite{SCZ17} considered inversions, insertions, and a further generalized operation where one removes a larger block of consecutive vertices from the ordering and inserts the block of consecutive vertices at a new position.

\begin{figure}
\begin{center}
\begin{tikzpicture}
\tikzstyle{knoten}=[circle,fill=white,draw=black,minimum size=6pt,inner sep=0pt]
\tikzstyle{bez}=[inner sep=0pt]


\draw[rounded corners, fill=black!10, draw=black!10] (-1.5, -0.7) rectangle (5, 0.5) {};
\node[bez] at (-1,0) {$\tau$};
\node[knoten,label=below:$v_1$] (v1) at (0,0) {};
\node[knoten,label=below:$v_2$] (v2) at (1,0) {};
\node[knoten,label=below:$v_3$] (v3) at (2,0) {};
\node[knoten,label=below:$v_4$] (v4) at (3,0) {};
\node[knoten,label=below:$v_5$] (v5) at (4,0) {};

\begin{scope}[xshift=-4cm,yshift=-3cm]
\draw[rounded corners, fill=black!10, draw=black!10] (-2, -1) rectangle (5, 1.5) {};
\node[bez] at (1.5,1) {Swap};
\node[bez] at (-1,0) {$\sigma_1$};
\node[knoten,label=below:$v_3$] (v1) at (0,0) {};
\node[knoten,label=below:$v_5$] (v2) at (1,0) {};
\node[knoten,label=below:$v_1$] (v3) at (2,0) {};
\node[knoten,label=below:$v_4$] (v4) at (3,0) {};
\node[knoten,label=below:$v_2$] (v5) at (4,0) {};

\draw[<->, line width=1pt, densely dotted, bend left=30]  (v1) to (v3);
\draw[<->, line width=1pt, densely dotted, bend left=20]  (v2) to (v5);
\end{scope}

\begin{scope}[xshift=4cm,yshift=-3cm]
\draw[rounded corners, fill=black!10, draw=black!10] (-2, -1) rectangle (5, 1.5) {};
\node[bez] at (1.5,1) {Inversions};
\node[bez] at (-1,0) {$\sigma_2$};
\node[knoten,label=below:$v_2$] (v1) at (0,0) {};
\node[knoten,label=below:$v_3$] (v2) at (1,0) {};
\node[knoten,label=below:$v_1$] (v3) at (2,0) {};
\node[knoten,label=below:$v_4$] (v4) at (3,0) {};
\node[knoten,label=below:$v_5$] (v5) at (4,0) {};

\draw[<->, line width=1pt, densely dotted, bend left=30]  (v1) to node[above]{\tiny{$1$}} (v2);
\draw[<->, line width=1pt, densely dotted, bend left=30]  (v2) to node[above]{\tiny{$2$}} (v3);
\end{scope}

\begin{scope}[xshift=-4cm,yshift=-6cm]
\draw[rounded corners, fill=black!10, draw=black!10] (-2, -1) rectangle (5, 1.5) {};
\node[bez] at (1.5,1) {Insert};
\node[bez] at (-1,0) {$\sigma_3$};
\node[knoten,label=below:\textcolor{black!50}{$v_1$}, dotted] (v1) at (0,0) {};
\node[knoten,label=below:$v_2$] (v2) at (1,0) {};
\node[knoten,label=below:$v_3$] (v3) at (2,0) {};
\node[knoten,label=below:$v_1$] (newv1) at (2.5,0) {};
\node[knoten,label=below:\textcolor{black!50}{$v_4$}, dotted] (v4) at (3,0) {};
\node[knoten,label=below:$v_5$] (v5) at (4,0) {};
\node[knoten,label=below:$v_4$] (newv4) at (4.5,0) {};

\draw[->, line width=1pt, densely dotted, bend left=20]  (v1) to (newv1);
\draw[->, line width=1pt, densely dotted, bend left=30]  (v4) to (newv4);
\end{scope}

\begin{scope}[xshift=4cm,yshift=-6cm]
\draw[rounded corners, fill=black!10, draw=black!10] (-2, -1) rectangle (5, 1.5) {};
\draw[rounded corners, fill=black!5, draw=black] (-0.3, -0.7) rectangle (2.3, 0.5) {};
\draw[rounded corners, fill=black!5, draw=black] (2.7, -0.7) rectangle (4.3, 0.5) {};
\node[bez] at (1.5,1) {Inversion-Window};
\node[bez] at (-1,0) {$\sigma_4$};
\node[knoten,label=below:$v_2$] (v1) at (0,0) {};
\node[knoten,label=below:$v_3$] (v2) at (1,0) {};
\node[knoten,label=below:$v_1$] (v3) at (2,0) {};
\node[knoten,label=below:$v_5$] (v4) at (3,0) {};
\node[knoten,label=below:$v_4$] (v5) at (4,0) {};

\draw[<->, line width=1pt, densely dotted, bend left=30]  (v1) to (v2);
\draw[<->, line width=1pt, densely dotted, bend left=30]  (v2) to (v3);

\draw[<->, line width=1pt, densely dotted, bend left=30]  (v4) to (v5);
\end{scope}

\end{tikzpicture}
\end{center}
\caption{Examples of the distances studied in this work. The upper part shows an ordering~$\tau$ of the vertices~$v_1, \dots, v_5$. The lower part shows orderings~$\sigma_1, \dots, \sigma_4$, that have distance~$2$ from~$\tau$ with respect to the swap-, inversions-, insert-, or inversion-window distance. The dotted arrows correspond to the operations performed to obtain the ordering~$\sigma_i$ from~$\tau$. In case of the swap distance,~$\sigma_1$ is obtained from~$\tau$ by swapping the vertex pairs marked in the figure. To obtain the ordering~$\sigma_2$, a first inversion swapped the positions of~$v_1$ and~$v_2$, and afterwards, a second inversion swapped the positions of~$v_1$ and~$v_3$. In case of the insert distance, $\sigma_3$ is obtained by removing the vertices~$v_1$ and~$v_4$ and inserting them at new positions. For the inversion-window distance, there are two windows such that at most two inversions are performed inside each window to obtain the ordering~$\sigma_4$ from~$\tau$.}\label{Figure: Distances}
\end{figure}
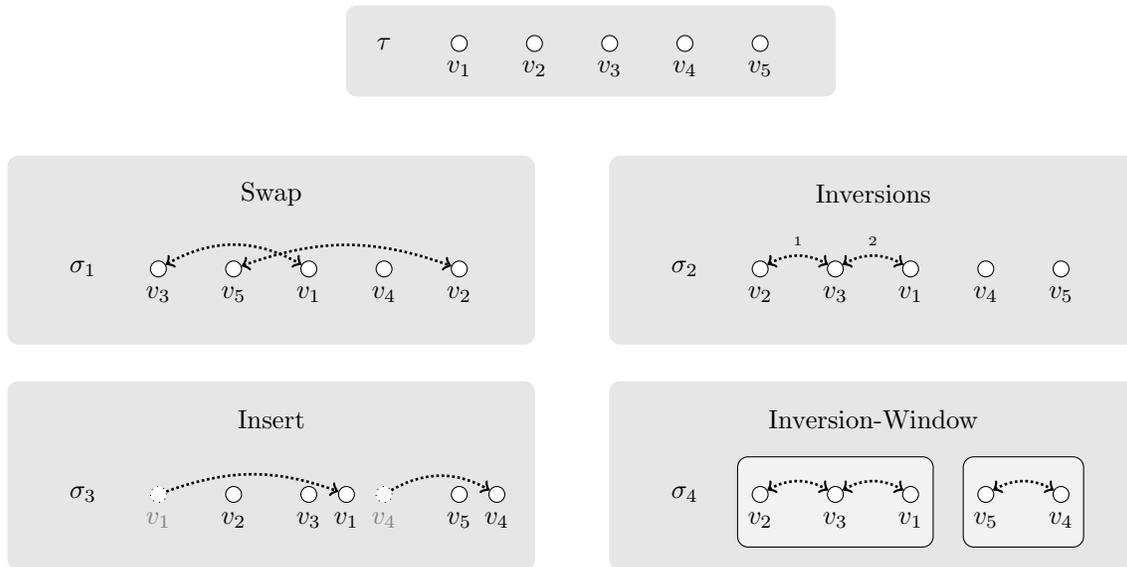

\paragraph{Our Results.} We study different versions of local search on variable orderings. In contrast to previous work that defined the local neighborhood of an ordering as all orderings that can be reached via \emph{one} operation, we consider \emph{parameterized} local search~\cite{FLRS10, MS10, GJOSS12, FFLRSV12, OS13}. Here, one sets a parameter~$r$ and aims to find a better network that can be reached via at most~$r$ modifications of the ordering. The hope is that, by considering such a larger neighborhood, one may avoid being stuck in a bad local optimum. We consider three different kinds of operations in this work. We first study \emph{insertions}, where one may move one variable to an arbitrary position in the ordering and \emph{swaps} where two arbitrary variables may exchange their positions. Recall that local search strategies based on one insertion have been studied previously~\cite{AOP11,SCZ17}. We observe that the local search problem can be solved in polynomial time for both variants if~$r$ is constant. The degree of the polynomial, however, depends on~$r$. We show that, under widely accepted complexity-theoretic assumptions, this dependence cannot be avoided. Afterwards, we study inversions, which are swaps of adjacent vertices. Our main result is an algorithm with running time~$2^{\Oh(\sqrt{r}\log r)}\cdot \poly(|I|)$ for deciding for a given variable ordering, whether there is a better ordering that can be reached via at most~$r$ inversions. The distance that measures the minimum number of inversions needed to transform one ordering in another is also known as the \emph{Kendall tau distance} as it is an adaption of Kedall tau rank correlation~\cite{K38}. We then introduce a new distance that we call \emph{inversion-window distance}. Intuitively, given an ordering~$\tau$ we find an ordering~$\tau'$ that has inversion-window distance at most~$r$ by partitioning~$\tau$ into multiple windows of consecutive vertices and performing up to~$r$ inversions inside each window. This new distance extends the number of inversions in the following sense: Given a search radius~$r$, the search space for orderings with inversion-window distance~$r$ is potentially larger than the number of orderings one can obtain with at mots~$r$ inversions. We provide an algorithm with running time~$2^{\Oh(\sqrt{r}\log r)}\cdot \poly(|I|)$ that decides for a given ordering, whether there is a better ordering that has inversion-window distance at most~$r$. An overview of the distances studied in this work is shown in Figure~\ref{Figure: Distances}.

Our algorithms work not only for the standard BNSL problem but also for some structural
constraints that we may wish to impose on the final DAG.  To formulate the algorithms compactly,
we introduce a generalization of BNSL, where each possible parent set
is associated with a score and a weight and the aim is to find the highest scoring network
that does not exceed a specified weight bound. We show that this captures natural types of
structural constraints.  As a side result, we show that a previous polynomial-time
algorithm for acyclic directed superstructures~\cite{OS13} can be generalized to this new more general
problem.

Finally, we show that for using an ordering-based local search approach in the presence of structural constraints it is essentially necessary, that the corresponding variant of BNSL is polynomial-time solvable on instances where the directed super structure is acyclic. This implies for several important structural constraints that ordering-based local search is unlikely to be useful.

\section{Preliminaries}
\paragraph{Notation.} We consider directed graphs~$D=(N,A)$ which consist of a vertex set~$N$ and an arc set~$A \subseteq N \times N$. If~$D$ contains no directed cycle, then~$D$ is called a \emph{directed acyclic graph (DAG)}. 
An arc~$(u,v) \in A$ is~\emph{outgoing from~$u$} and~\emph{incoming into~$v$}. 
A \emph{source} is a vertex without incoming arcs, and a \emph{sink} is a vertex without outgoing arcs. The set~$P^A_v:=\{u \in N \mid (u,v) \in A\}$ is the~\emph{parent set of~$v$}, and for every~$u \in P^A_v$,~$v$ is called \emph{child of~$u$}. Throughout this work let~$n:=|N|$. Given~$f:X \rightarrow Y$ and~$X' \subseteq X$, we let~$f|_{X'}:X' \rightarrow Y$ denote the \emph{limitation of~$f$ to~$X'$} which is defined by~$f|_{X'}(x):= f(x)$ for every~$x \in X'$. We let~$[i,j]$ denote the set of integers~$\ell$ such that~$i\le \ell\le j$.

\paragraph{Orderings.}  Given a vertex set~$N$, an \emph{ordering of~$N$} is an~$n$-tuple~$\tau=(v_1, \dots, v_n)$ containing every vertex of~$N$. For~$i \leq n$, we let~$\tau(i)$ denote the~$i$th vertex appearing on~$\tau$. We write~$u <_\tau v$ if the vertex~$u$ appears before~$v$ on~$\tau$. A~\emph{partial ordering of~$\tau$} is an ordering~$\sigma$ of a subset~$S \subseteq N$ such that~$u <_\tau v$ if and only if~$u <_\sigma v$ for all~$u,v \in S$. Given a vertex set~$S$, we let~$\tau[S]$ denote the partial ordering containing exactly the vertices from~$S$, and we let~$\tau - S:= \tau[N \setminus S]$ denote the partial ordering we obtain when removing all vertices of~$S$ from~$\tau$. For~$i \in [1,n]$ and~$j\in [i,n]$ we define~$\tau(i,j)$ as the partial ordering~$(\tau(i),\tau(i+1),\dots,\tau(j))$.
 Given a partial ordering~$\sigma$ of~$\tau$, we let~$N(\sigma)$ be the set of all elements appearing on~$\sigma$. 
 Moreover, we may use the abbreviation~$N_\tau(i,j):=N(\tau(i,j))$. 

A \emph{distance~$d$} \iflong of the orderings of~$N$ \fi is a mapping that assigns an integer to every pair of orderings~$\tau$ and~$\tau'$ of~$N$ such that~$d(\tau,\tau')= d(\tau', \tau)$ and~$d(\tau,\tau)=0$. For an integer~$r$, we say that an ordering~$\tau'$ is~\emph{$r$-close to~$\tau$ with respect to~$d$} if~$d(\tau, \tau') \leq r$. If~$\tau$ and~$d$ are clear from the context we may only write \emph{$\tau'$ is~$r$-close}.

 Let~$D:=(N,A)$ be a directed graph. An ordering~$\tau$ of~$N$ is a \emph{topological ordering of~$D$} if~$u<_{\tau}w$ for every arc~$(u,w) \in A$. A directed graph has a topological ordering if and only if it is a DAG.

\paragraph{Parameterized Complexity.}  A problem is \emph{slicewise polynomial} (\XP) for some parameter~$k$ if it can be solved in time~$|I|^{f(k)}$ for a computable function~$f$. That is, the problem is solvable in polynomial~time when~$k$ is constant. A problem is called~\emph{fixed-parameter tractable}~(\FPT) for a parameter~$k$ if it can be solved in time~$f(k) \cdot |I|^{\Oh(1)}$ for a computable function~$f$.
In general,~\FPT running times are preferable, because the degree of the polynomial is independent from the value of~$k$, whereas in the case of~\XP running times, the degree of the polynomial might depend on the value of~$k$.
 If a problem is \W1-hard, then it is assumed to be fixed-parameter \emph{intractable}. 
 The \textsc{Clique} problem when parameterized by size of the sought clique is an example for a \W1-hard problem. 
 For a detailed introduction into parameterized complexity we refer to the standard monographs~\cite{CFKLMPPS15,DF13}.

\section{Bayesian Network Structure Learning with Multiscores}
In this work we provide local search algorithms for \textsc{Bayesian Network Structure Learning (BNSL)} where one aims to find a DAG that maximizes some score. 
The presented algorithms also work for some generalizations of BNSL where one aims for example to find a DAG with a restricted number of arcs. To capture these generalizations we introduce the problem~\textsc{Weighted Bayesian Network Structure Learning (W-BNSL)}. 
 
Let~$N$ be a set of vertices. A mapping~$\Fa$ is a collection of \emph{local multiscores for~$N$} if~$\Fa(v) \subseteq 2^{N \setminus \{v\}} \times \mathds{N}_0 \times \mathds{N}_0$ for each~$v\in N$. 
Intuitively, if~$(P,s,\w)\in \Fa(v)$ for some vertex~$v \in N$, then choosing~$P$ as the parent set of~$v$ may simultaneously give a local score of~$s$ and a local weight of~$\w$.
For the same parent set~$P$, there might be another local multiscore~$(P,s',\w')\in \Fa(v)$ with a different local score and a different local weight.
Throughout this work we assume that for every vertex~$v$ there exists some~$s\in \mathds{N}_0$ such that~$(\emptyset, s,0) \in \Fa(v)$, that is, every vertex has a local multiscore for the empty parent set with weight zero. 
Given~$v \in N$ and~$\Fa$, the \emph{possible parent sets} are defined by~$\Pa_\Fa(v):=\{P \mid \exists s\in \mathds{N}_0:\exists \w \in \mathds{N}_0: (P,s,\w) \in \Fa(v)\}$. 
  Given~$N$ and~$\Fa$, the directed graph~$S_\Fa:= (N,A_\Fa)$ with~$A_\Fa:=\{(u,v) \mid \exists P \in \Pa_\Fa(v): u \in P\}$ is called the~\emph{superstructure} of~$N$ and~$\Fa$~\cite{OS13}.
  
\begin{definition} \label{Def Reasonable Arc Set}
Let~$N$ be a vertex set, let~$\Fa$ be multiscores for~$N$. 
An arc set~$A \subseteq N \times N$ is called~\emph{$\Fa$-valid} if~$(N,A)$ is a DAG and~$P^A_v \in \Pa_\Fa(v)$ for all~$v\in N$.
\end{definition}

We say that an~$\Fa$-valid arc set has weight at most~$k$ if for every~$v \in N$ there is some~$(P^A_v,s_v,\w_v) \in \Fa(v)$ such that~$\sum_{v\in N} \w_v \leq k$.
For a given integer~$k$ and an~$\Fa$-valid arc set~$A$ we define~$\score_{\Fa}(A,k):=\sum_{v \in N} s_v$ as the maximal score one can obtain from any choice of triples~$(P^A_v, s_v,\w_v) \in \Fa_v, v\in N$, with~$\sum_{v\in N} \w_v \leq k$.  If~$\Fa$ is clear from the context we may write~$\score(A,k):=\score_\Fa(A,k)$. We now formally define the problem. 

\prob{Weighted Bayesian Network Structure Learning (\WBNSL)}
{A set of vertices~$N$, local multiscores~$\Fa$, and two integers~$t,k \in \mathds{N}_0$.}
{Is there an $\Fa$-valid arc set~$A \subseteq N \times N$ such that~$\score(A,k) \geq t$?}

For~$N=\{v_1, \dots, v_n\}$ the local multiscores~$\Fa$  are given as a two-dimensional array~$\Fa:=[Q_1, \dots, Q_n]$, where each~$Q_i$ is an array containing a quadruple~$(s,\w, |P|, P)$ for each~$(P,s,\w)\in \Fa(v)$. The instance size~$|I|$ is the number of bits needed to store this two-dimensional array. Throughout this work we assume that~$k \in \poly(|I|)$ for every instance~$I:=(N,\Fa,t,k)$ of~\WBNSL. 
Note that this also implies~$\w \in \poly(|I|)$ for each multiscore~$(P,s,\w)$. 

\WBNSL generalizes the classical \textsc{Bayesian Network Structure Learning (BNSL)}. In \textsc{BNSL} one is given a set~$N$ of vertices and one local score~$s$ for each pair consisting of a vertex~$v \in N$ and a possible parent set~$P \subseteq N \setminus \{v\}$ and the goal is to learn a DAG such that the sum of the local scores is maximal. 
Thus, \textsc{BNSL} can be modeled with local multiscores~$\Fa(v)$ containing triples~$(P,s,0)$. 
Since \textsc{BNSL} is NP-hard~\cite{C95} and the weights~$\omega$ are not used for this construction, \WBNSL is NP-hard even if~$k=0$.

\WBNSL also allows to model the task of Bayesian network structure learning under additional sparsity constraints: Recall that one example for such a constrained version is \textsc{Bounded Arcs-BNSL (BA-BNSL)}. In \textsc{BA-BNSL} one aims to learn a DAG that consists of at most~$k$ arcs for some given integer~$k$. This can be modeled with multiscores containing triples~$(P,s,|P|)$. \textsc{BA-BNSL} is known to be fixed-parameter tractable when parameterized by the number~$k$ of arcs~\cite{GK20}.

A further example is~\textsc{Bounded Indegree $c$-BNSL (BI-$c$-BNSL)} which is defined for every constant~$c$. 
In \textsc{BI-$c$-BNSL} one aims to learn a network that contains at most $k$~vertices that have more than~$c$ parents for a given integer~$k$. 
This scenario can be modeled with triples~$(P,s,\w)$ where~$\w=1$ if~$|P| > c$ and~$\w=0$, otherwise. 
To the best of our knowledge, this is a sparsity constraint that has not been analyzed so far. 
Next, we observe that \WBNSL is solvable in polynomial time if the superstructure is a DAG. This generalizes algorithms for~\textsc{BNSL}~\cite{OS13} and~\textsc{BA-BNSL}~\cite{GK20}. 
\begin{theorem}\label{thm poly if dag}
\WBNSL is solvable in $\Oh(k\cdot |I|)$~time if~$S_\Fa$ is a DAG.
\end{theorem}
\begin{proof}
Let~$I=(N,\Fa, t, k)$ be an instance of~\WBNSL where the superstructure~$S_\Fa$ is a DAG and let~$\tau$ be a topological ordering of~$S_\Fa$.
We describe a dynamic programming algorithm to solve~$I$.
 The dynamic programming table~$T$ has entries of type~$T[i,k']$ with~$i \in [1, n+1]$ and $k' \in [0,k]$. 
 Each entry stores the maximal score of an arc set~$A$ of weight at most~$k'$ where only the vertices of~$N_\tau(i,n)$ are allowed to learn a non-empty parent set and only the local multiscores of the vertices of~$N_\tau(i,n)$ count towards the score and weight of~$A$.
 We start to fill the table~$T$ by setting~$T[n+1,k']:= 0$ for all~$k'\in [0,k]$.
 The recurrence to compute an entry for~$i\in[1,n]$ and~$k'\in[0,k]$ is
 \begin{align*}
T[i,k'] := \max_{\substack{(P,s,\w) \in \Fa(v)\\\w \leq k'}} s + T[i+1, k' - \w].
\end{align*}
where~$v:=\tau(i)$.
Thus, to determine if~$I$ is a yes-instance of~\WBNSL, it remains to check if~$T[1,k] \geq t$.
The corresponding network can be found via traceback.
The formal correctness proof is straightforward and thus omitted.

The dynamic programming table~$T$ consists of~$k+1$ entries for each~$v\in V$ and each such entry can be computed in $\Oh(|\Fa(v)|)$~time plus~$k+1$ entries which can be computed in~$\Oh(1)$~time. 
Hence, the total running time is~$\Oh(k \cdot |I|)$. 
Consequently, the algorithm runs in polynomial time since~$k \in \poly(|I|)$.
\end{proof}

Assume we are given an instance~$I:=(N,\Fa,t,k)$ of~\textsc{W-BNSL} and consider an arbitrary but fixed optimal solution~$A$ for~$I$.
Suppose, we are given some topological ordering~$\tau$ of~$(N,A)$ but not the arc set~$A$ itself. 
Then, Theorem~\ref{thm poly if dag} implies that we can solve~$I$ in polynomial time by restricting the possible parent sets to parent sets respecting the ordering~$\tau$. 
This gives rise to the ordering-based local search approach which we study in this work. 
More precisely, we consider a version of \WBNSL where one is additionally given an ordering of the vertex set and an integer~$r$ and aims to learn a DAG that has a topological ordering that is~$r$-close to the given ordering. 
For any fixed distance~$d$, this problem is formally defined as follows.
\begin{center}
	\begin{minipage}[c]{.9\linewidth}
          \textsc{$d$-Local W-BNSL}\\
          \textbf{Input}: 
A set of vertices~$N$, local multiscores~$\Fa$, an ordering~$\tau$ of~$N$, and three integers~$t, k,r \in \mathds{N}$.\\
          \textbf{Question}: Is there an~$\Fa$-valid arc set~$A$ such that~$\score(A,k) \geq t$ and~$(N,A)$ has a topological ordering~$\tau'$ that is~$r$-close to~$\tau$ with respect to~$d$?
	\end{minipage}
\end{center}

Let~$I$ be an instance of~\textsc{$d$-Local W-BNSL}.
We call an arc set~$A$ \emph{feasible for~$I$} if~$A$ is~$\Fa$-valid with weight at most~$k$ and~$(N,A)$ has a topological ordering that is~$r$-close to~$\tau$. For theoretical reasons~\textsc{$d$-Local W-BNSL} is stated as a decision problem since this allows us to prove the presented conditional lower bounds for some cases of~$d$. 
However, the algorithms presented in this work solve the corresponding optimization problem, where the input is an instance~$I:=(N,\Fa,\tau,k,r)$ and the task is to compute a feasible arc set that maximizes~$\score(A,k)$. 
Throughout this work, such an arc set is called a \emph{solution of~$I$}.
Note that not every feasible arc set is necessarily a solution.

\section{Preliminary Experiments} \label{Section: Experiments}

To assess whether parameterized local search is in
principle a viable approach to ordering-based Bayesian network
structure learning, we perform some preliminary experiments for the
standard BNSL problem with a particularly simple neighborhood.
Consider the \emph{window-distance}~$\Win$ where~$\Win(\tau,\tau')$ is the
distance~$r$ between the first and last position in which~$\tau$
and~$\tau'$ differ. Formally, for two orderings of length~$n$, we
define~$\Win(\tau,\tau'):=0$ if~$\tau=\tau'$ and otherwise
$$\Win(\tau,\tau'):=\max_{j\in [1,n]: \tau(j)\neq \tau'(j)} j -
\min_{i\in [1,n]: \tau(i)\neq \tau'(i)} i.$$

We now use a hill-climbing algorithm that combines the
\emph{window-distance} with insertion operations. That is, we say that
two orderings~$\tau$ and~$\tau'$ are~\emph{$r$-close}
if~$\Win(\tau,\tau')\le r$ or~$\tau$ can be obtained from~$\tau'$ via
one insertion operation. We say that an ordering~$\tau$ is~\emph{$r$-optimal}
if there is no ordering~$\tau'$ that is~$r$-close to~$\tau$ such that
the best DAG with topological ordering~$\tau'$ achieves a better score
than the best DAG with  topological ordering~$\tau$.

The algorithm to compute an~\emph{$r$-optimal} ordering works as
follows; see Figure~\ref{fig:pseudo} for the pseudocode. Given a start
ordering~$\tau$, we repeat the following steps until no further
improvement was found: First, check if some single insert operation
on~$\tau$ gives an improvement and apply it if this is the case.
Otherwise, slide a window of size~$r$ over the ordering~$\tau$ and
find a permutation for this window that optimizes the score of the
total ordering. Apply the permutation to the window if it leads to an
improved score. Repeat these two steps until no further improvement
has been found. To find the optimal permuation of the window one may
try all permutations of the window entries. We use a more efficient
dynamic programming algorithm; since we were not too interested in a
detailed running time evaluations in this preliminary experiment, we
omit further details on this dynamic programming algorithm.
\begin{figure}[t]
  
  \begin{algorithm}[H]\small
        \SetKwInOut{Input}{Input}
        \SetKwInOut{Output}{Output}
        \SetKw{Break}{break}
        \Input{A set of vertices~$N$, local multiscores~$\mathcal{F}$, an ordering~$\tau$ of~$N$ and an integer~$r$.}
        
        \Output{An $r$-optimal ordering~$\tau$.}

         $\CurrentScore := \score_{\mathcal{F}}(\tau)$\\
        $\improvable := $ true\\
        \While{$\improvable$}{
          $\improvable \leftarrow$ false\\
          \ForEach{$\tau'$ such that~$\tau$ can be obtained from~$\tau'$ via one insertion}
          {
            \If{$\CurrentScore < \score_{\mathcal{F}}(\tau')$}
            {$\CurrentScore \leftarrow \score_{\mathcal{F}}(\tau')$; $\tau\leftarrow \tau'$\\
              $\improvable \leftarrow$ true\\
              \Break
            }
          }
          \If{$\mathrm{not} \improvable$}{
            \For{$i= 1$ \KwTo $n-r$}{
              $\tau_i:= $ permutation of $\tau[i,i+r]$ such that~$\score_{\mathcal{F}}(\tau[1,i-1]\circ \tau_i\circ \tau[i+r+1,n])$ is maximum.\\
              $\tau':= \tau[1,i-1]\circ \tau_i\circ \tau[i+r+1,n]$\\
              \If{$\CurrentScore < \score_{\mathcal{F}}(\tau')$}
              {$\CurrentScore \leftarrow \score_{\mathcal{F}}(\tau')$; $\tau\leftarrow \tau'$\\
                $\improvable \leftarrow$ true\\
                \Break
                }
            }
          }
        }
    \Return $\tau$
  \end{algorithm}
  
  \caption{A local search algorithm combining two simple neighborhoods, herein~$\score_{\mathcal{F}}(\tau)$ for an ordering~$\tau$ denotes maximum score of any DAG~$D$ such that~$\tau$ is a topological ordering of~$D$.}
  \label{fig:pseudo}
\end{figure}
We performed an experimental evaluation of this algorithm on data sets provided at the GOBNILP~\cite{CB13}
homepage.\footnote{\url{https://www.cs.york.ac.uk/aig/sw/gobnilp/}}
Given an instance, we compute~$20$ random topological orderings. We ran experiments for
each~$r\in \{3,5,7,9, 11\}$ using the same 20 random orderings for a fair comparison.  

Table~\ref{Table: Experiments} shows for each instance and each~$r$ the average score of the computed~$r$-optimal orderings; Table~\ref{Table: Experiments2} shows the maximum score of the 20 computed~$r$-optimal orderings. For most of the
instances the best results are obtained for~$r\in \{9,11\}$ in terms of both average score and maximum score. Thus, the experiments show that it can be worthwhile to consider larger local search neighborhoods, demonstrated here for the combination of window distance with parameter~$r$ and insertion operation. It is notable that we see the positive effect of a larger search radius for the window distance even when the search neighborhood allows for another non-window operation. This shows that the positive effect of the larger search radius is not due to choosing a too restrictive search neighborhood in the first place. Finally, we remark that the running time bottleneck in our
preliminary experiments was not the combinatorial explosion in~$r$ but rather the slow implementation of the insert operation. 
 \begin{table}
	\caption{Results of the experiments for local search with inserts and with inversions inside a window of size~$r$. The boldface entries mark the maximal average score obtained for each instance.}
  \small \centering  
  \begin{tabular}{lrrrrr}
    
    \hline
    
    instance & $r=3$, avg & $r=5$, avg & $r=7$, avg & $r=9$, avg & $r=11$, avg \\
    \hline
    
alarm-10000 & -105308.81 & \textbf{-105300.16} & \textbf{-105300.16} & \textbf{-105300.16} & \textbf{-105300.16} \\

alarm-1000 & -11265.42 & -11265.16 & -11264.91 & -11264.96 & \textbf{-11264.60}\\

alarm-100 & -1357.26 & -1357.26 & -1357.26 & -1356.33 & \textbf{-1355.77}\\

asia-10000 & -22467.52 & \textbf{-22466.40} & \textbf{-22466.40} & -22467.52 & -22467.52\\

asia-1000 & \textbf{-2317.49} & \textbf{-2317.49} & \textbf{-2317.49} & \textbf{-2317.49} & \textbf{-2317.49} \\

asia-100 & -246.96 & -246.30 & \textbf{-245.81} & -246.96 & -246.96\\

carpo-10000 & -174451.06 & -174450.96 & -174450.96 & -174404.82 & \textbf{-174397.24}\\

carpo-1000 & -17747.31 & -17745.21 & -17746.77 & -17746.31 & \textbf{-17739.14}\\

carpo-100 & -1844.54 & -1844.54 & -1844.54 & -1843.20 & \textbf{-1842.42} \\

hailfinder-10000 & -498133.54 & -498133.54 & -498133.54 & \textbf{-498099.56} & \textbf{-498099.56}\\

hailfinder-1000 & -52508.02 & -52508.02 & -52508.02 & -52508.02 & \textbf{-52507.98}\\

hailfinder-100 & \textbf{-6021.58} & \textbf{-6021.58} & \textbf{-6021.58} & \textbf{-6021.58} & \textbf{-6021.58}\\

insurance-10000 & -133108.70 & -133086.47 & -133086.47 & \textbf{-133055.32} & \textbf{-133055.32}\\

insurance-1000 & -13931.21 & -13929.06 & -13924.03 & -13915.04 & \textbf{-13914.44}\\

    insurance-100 & -1695.65 & -1695.46 & -1694.27 & \textbf{-1693.40} & -1693.91\\
    
    kredit-family & -16702.23 & -16702.23 & -16700.58 & \textbf{-16697.58} & -16698.53 \\
    
 water-1000 & -13274.46 & -13274.46 & -13274.46 & -13274.16 & \textbf{-13270.91}\\
 
    water-100 & -1502.62 & -1502.62 & -1502.40 & -1502.14 & \textbf{-1501.85}\\
    
\hline
  \end{tabular}
  \label{Table: Experiments}
  \end{table}

  \begin{table}
  \caption{	Results of the experiments for local search with inserts and with inversions inside a window of size~$r$. The boldface entries mark the maximal maximum score obtained for each instance.}

    \begin{center}

  \begin{tabular}{lrrrrr}
  
      \hline
    
    instance & max & $r=5$, max & $r=7$,max & $r=9$, max & $r=11$, max  \\
    \hline
    
alarm-10000 & \textbf{-105226.51} & \textbf{-105226.51} & \textbf{-105226.51} & \textbf{-105226.51} & \textbf{-105226.51}\\

alarm-1000 &  \textbf{-11247.28} & \textbf{-11247.28} & \textbf{-11247.28} & \textbf{-11247.28} & \textbf{-11247.28}\\

alarm-100 & -1351.92 & -1351.92 & -1351.92 & -1351.01 & \textbf{-1350.55}\\

asia-10000 & \textbf{-22466.40} & \textbf{-22466.40} & \textbf{-22466.40} & \textbf{-22466.40} & \textbf{-22466.40}\\

asia-1000 &  \textbf{-2317.41} & \textbf{-2317.41} & \textbf{-2317.41} & \textbf{-2317.41} & \textbf{-2317.41}\\

asia-100 &  \textbf{-245.64} & \textbf{-245.64} & \textbf{-245.64} & \textbf{-245.64} & \textbf{-245.64}\\

carpo-10000 & -174269.97 & -174269.97 & -174269.97 & \textbf{-174137.03} & -174139.69\\

carpo-1000 & -17728.98 & -17728.98 & -17725.29 & -17724.56 & \textbf{-17724.05}\\

carpo-100 &  \textbf{-1839.16} & \textbf{-1839.16} & \textbf{-1839.16} & \textbf{-1839.16} & \textbf{-1839.16}\\

hailfinder-10000 & \textbf{-497730.35} & \textbf{-497730.35} & \textbf{-497730.35} & \textbf{-497730.35} & \textbf{-497730.35}\\

hailfinder-1000 &  \textbf{-52486.74} & \textbf{-52486.74} & \textbf{-52486.74} & \textbf{-52486.74} & \textbf{-52486.74}\\

hailfinder-100 & \textbf{-6019.47} & \textbf{-6019.47} & \textbf{-6019.47} & \textbf{-6019.47} & \textbf{-6019.47}\\

insurance-10000 & \textbf{-132968.58} & \textbf{-132968.58} & \textbf{-132968.58} & \textbf{-132968.58} & \textbf{-132968.58}\\

insurance-1000 & -13909.50 & -13888.03 & -13888.03 & \textbf{-13887.90} & -13888.58\\

    insurance-100 & -1689.90 & -1689.90 & -1689.90 & -1689.24 & \textbf{-1689.17}\\
    
    kredit-family & \textbf{-16695.67} & \textbf{-16695.67} & \textbf{-16695.67} & \textbf{-16695.67} & \textbf{-16695.67}\\
    
 water-1000 & \textbf{-13263.38} & \textbf{-13263.38} & \textbf{-13263.38} & \textbf{-13263.38} & \textbf{-13263.38}\\
 
    water-100 & -1501.26 & -1501.26 & -1501.26 & -1501.19 & \textbf{-1501.03}\\
    
\hline
  
  \end{tabular}
  \label{Table: Experiments2}
  \end{center}
\end{table}

\section{Parameterized Local Search for Insert and Swap Distances}
A \emph{swap operation} on two vertices~$v$ and~$w$ on an ordering~$\tau$ interchanges the positions of~$v$ and~$w$. The distance~$\text{Swap}(\tau,\tau')$ is the minimum number of swap operations needed to transform~$\tau$ into~$\tau'$.
An \emph{insert operation} on an ordering~$\tau$ removes one arbitrary vertex from~$\tau$ and inserts it at a new position. We define
$\text{Insert}(\tau,\tau')$ as the minimum number of insert operations needed to transform~$\tau$ into~$\tau'$.
This number can be computed as~$\text{Insert}(\tau,\tau') = |N| - \text{LCS}(\tau, \tau')$, where LCS$(\tau, \tau')$ is the length of the longest common subsequence of~$\tau$ and~$\tau'$.
That is, if~$\text{Insert}(\tau,\tau') = r$, then there is a subset~$S\subseteq V$ of size~$r$ such that~$\tau - S = \tau' - S$ and vice versa.

For both distances, local search approaches for BNSL have been studied previously~\cite{AOP11,LB17,SCZ17}. We now focus on the parameterized complexity regarding the parameter~$r$ which is the radius of the local search neighborhood. 
We first prove that \IV and \SW are \XP when parameterized by~$r$.
That is, both problems are solvable in polynomial time if~$r$ is a constant. 
The following algorithm is straight forward and simply computes for each possible ordering~$\tau'$ which is~$r$-close to~$\tau$ a feasible arc set~$A$ such that~$\tau'$ is a topological ordering of~$(N,A)$ and~$\score(A,k)$ is maximal.
The latter is done by applying~\Cref{thm poly if dag} after restricting the local multiscores to those that respect the ordering~$\tau'$.

\begin{theorem} \label{Theorem: XP for Insert}
\IV and \SW are solvable in~$n^{\Oh(r)}\cdot \poly(|I|)$~time.
\end{theorem}

\iflong
\begin{proof} 
Given an instance~$I=(N,\Fa, \tau, t, k, r)$ of~
\textsc{Insert-Local \WBNSL}, we compute all subsets~$S\subseteq V$ of size~$r$ and all orderings~$\tau'$ with~$\tau' - S = \tau - S$ in~$n^{\Oh(r)}$ time.
For each such~$\tau'$, we compute the instance~$I'=(N, \Fa', t, k)$ of~\textsc{\WBNSL}, where~$\Fa'(v) := \{(P,s,\w)\in \Fa \mid \forall u\in P: u<_{\tau'} v\}$ for each~$v\in N$.
Intuitively,~$\Fa'$ is the collection of local multiscores restricted to parent sets that respect the topological ordering~$\tau'$.
Due to Theorem~\ref{thm poly if dag}, each of these instances can be solved in polynomial time since~$S_{\Fa'}$ is a DAG.
If one of these instances is a yes-instance, then there is a set of arcs~$A\subseteq N\times N$ such that~$\score_{\Fa}(A,k) \geq \score_{\Fa'}(A,k) \geq t$.
Consequently,~$I$ is a yes-instance of~
\textsc{Insert-Local \WBNSL}.
The converse clearly holds as well since we iterate over all possible choices of~$S$ and~$\tau'$.

The algorithm for \SW works analogously: We iterate over all~$n^{\Oh(r)}$ collections of~$r$ vertex pairs that swap their positions. For each such choice we restrict the local multiscores with respect to the corresponding ordering and apply the algorithm behind Theorem~\ref{thm poly if dag}.
\end{proof}
\fi
Note that the running time is polynomial for constant values of~$r$ but the degree of the polynomial depends on~$r$.
Hence, using this algorithm is impractical even for small values of~$r$. 
Next we show that there is little hope that this algorithm can be improved to a fixed-parameter algorithm by showing that both problems are \W1-hard when parameterized by~$r$.

\begin{theorem} \label{Theorem: W1-h}
\IV and \SW are \W1-hard when parameterized by~$r$ even if~$k = 0$,~$|\Fa(v)| \leq 2$ for all~$v\in N$,~$S_\Fa$ is a DAG, and every potential parent set has size at most two.
\end{theorem}

\newcommand{\proofWhardness}{
$(\Rightarrow)$
Suppose that~$I$ is a yes-instance of~\textsc{Clique}.
Hence, there is a clique~$S\subseteq V$ of size~$k$ in~$G$.
Let~$\tau'$ be the permutation of~$N$ obtained by moving all the vertices of~$S$ in an arbitrary ordering to the beginning of~$\tau$. 
Thus,~$\text{Insert}(\tau, \tau') = k = r$.
Since~$S$ is a clique in~$G$, it follows that~$\{u,v\}\in E$ for all distinct vertices~$u,v\in S$.
Thus, we set~$A = (\widehat{A} \setminus \{(x_k, v)\mid v\in S\})\cup \{(u,w_i),(v,w_i)\mid e_i=\{u,v\}\in E(S)\}$. 
By construction,~$\tau'$ is a topological ordering for~$D=(N,A)$.
Moreover, it holds that~$\score(A)-\score(\widehat{A}) = |E(S)|\cdot k - k({k \choose 2}-1) ={k \choose 2} k - k({k \choose 2}-1) = k$ and, thus,~$\score(A) = t$.
Consequently,~$I'$ is a yes-instance of~\IV.

$(\Leftarrow)$
Suppose that~$I'$ is a yes-instance of~\IV.
Consequently, there is a set of arcs~$A\subseteq N\times N$ and a topological ordering~$\tau'$ for~$D=(N,A)$ such that~$\score(A) \geq  t = \score(\widehat{A}) + k$ and~$\text{Insert}(\tau, \tau')\leq r=k$.
By construction, the total possible score one can obtain from any arc set is at most~$n({k\choose 2}-1) + |A^*|n^9 + mk$. 
It follows that~$A^*\subseteq A$ as, otherwise,~$\score(A) \leq n({k\choose 2}-1) + (|A^*|-1)n^9 + mk < n({k\choose 2}-1) + |A^*|n^9 = \score(\widehat{A})$.

By the fact that~$\text{Insert}(\tau, \tau')\leq k$, there is a set of vertices~$S\subseteq N$ of size at most~$k$ such that~$\tau - S = \tau' - S$.
Since~$\score(A) \geq \score(\widehat{A}) + k$, it follows that there is a nonempty set of edges~$E'\subseteq E$ of~$G$ such that~$(u, w_i)\in A$ and~$(v, w_i)\in A$ for all~$e_i=\{u,v\}\in E'$.
Hence,~$w_i \in S$ or~$\{u,v\}\subseteq S$ for each~$e_i=\{u,v\}\in E'$ and thus~$S\cap (V \cup \{w_i\mid e_i\in E\}) \neq \emptyset$. 
We set~$S':= \{u,v\mid \{u,v\}\in E'\}$ and show that~$S' = S$ and that~$S'$ forms a clique of size~$k$ in~$G$. 
To this end, we show that~$(x_k, v)\not\in A$ for every~$v\in S'$.

Let~$e_i=\{u,v\}\in E'$. 
From the fact that~$(u, w_i)\in A$ and~$(v, w_i)\in A$, it follows that~$u<_{\tau'} w_i$ and~$v<_{\tau'} w_i$. 
Moreover, since~$A^*\subseteq A$ it holds that~$x_1<_{\tau'}x_2<_{\tau'}\dots<_{\tau'}x_k$ and~$w_i<_{\tau'}w_i^j$ for every~$j\in[1,k]$.
Hence,~$u<_{\tau'}x_k$ and~$v<_{\tau'}x_k$ as, otherwise,~$\{w_i^j\mid j\in[1,k]\}\subsetneq S$ or~$\{x_j\mid j\in[1,k]\}\subsetneq S$ which is impossible since~$|S|\leq k$ and~$S\cap (V \cup \{w_i\mid e_i\in E\}) \neq \emptyset$.
Consequently,~$(x_k, v)\not\in A$ for every~$v\in S'$.

Hence,~$\score(A)-\score(\widehat{A}) = |E'|\cdot k - |S'|({k \choose 2} - 1) \leq |E_G(S')|\cdot k - |S'|({k \choose 2} - 1)$.
Since~$\score(A) \geq \score(\widehat{A})+k$ and~$|S'|\leq |S|\leq k$ it follows that~$k \leq |E_G(S')|\cdot k - |S'|({k \choose 2} - 1)$ which is only possible if~$|S'| = k$ and~$|E_G(S')|={k \choose 2}$, since~$E_G(S') \subseteq {{S'}\choose 2}$.
Consequently,~$S'$ forms a clique of size~$k$ in~$G$ and, thus,~$I$ is a yes-instance of~\textsc{Clique}.
}

\begin{proof}
We describe a parameterized reduction from \textsc{Clique}. In \textsc{Clique} one is given an undirected graph~$G$ together with an integer~$k$ and the question is if~$G$ contains a clique of size~$k$, that is, a set of pairwise adjacent vertices. \textsc{Clique} is \W1-hard when parameterized by~$k$~\cite{DF13}.

Given an instance~$I=(G=(V,E),k)$ of~\textsc{Clique} where~$G$ has $n$~vertices and $m$~edges, we compute an equivalent instance~$I'=(N,\Fa, \tau, t, k', r)$ of~\IV in polynomial time.
Let~$V=\{v_1, \dots, v_n\}$, and let~$E=\{e_1, \dots, e_m\}$.
We start with an empty vertex set~$N$ and add for every edge~$e_i\in E$ the vertices~$w_i$ and~$\{w_i^1, \dots, w_i^k\}$. 
Moreover, we add the vertices of~$V$ and~$k$ additional vertices~$x_1, \dots, x_k$ to~$N$.

For every~$e_i\in E$, we set~$\Fa(w_i):=\{(e_i, k, 0)\}$ and~$\Fa(w^j_i):=\{(\{w_i\}, n^9, 0)\}$ with~$j\in[1,k]$. 
Further, we set~$\Fa(x_i):=\{(\{x_{i-1}\},n^9,0)\}$ for all~$i\in[2,k]$ and~$\Fa(v):=\{(\{x_{k}\},{k\choose 2}-1,0)\}$ for every vertex~$v\in V$.
Moreover, for each~$v\in N$, we also add~$(\emptyset,0,0)$ to~$\Fa(v)$.
Note that~$S_\Fa$ is a DAG.

Next, we describe the ordering~$\tau$ of the instance~$I'$. 
For each~$j\in [1,m]$, we set~$\tau_j:= (w_j, w_j^1, w_j^2, \dots, , w_j^k)$ and~$\tau := \tau_1 \cdot \tau_2 \cdot \ldots \cdot \tau_m \cdot(x_1, \dots, x_k, v_1, \dots, v_n).$

Finally, we set~$r:=k$,~$k':= 0$, and~$t := (mk + k-1) n^9 + n({k \choose 2}-1) + k$ which completes the construction of~$I'$.

Since each parent set has weight zero, we abbreviate in the following~$\score(X) := \score(X,k')$ for every~$X\subseteq N \times N$.
Let~$$A^* := \{(w_i,w_i^j)\mid e_i\in E, j\in[1,k]\} \cup \{(x_{i-1}, x_i)\mid i\in[2,k]\}$$ be the set of arcs for parent sets of score~$n^9$ and let~$\widehat{A}:=A^* \cup \{(x_k, v_i)\mid v_i \in V\}$ be the set of all arcs of parent sets with positive score that do not violate the topological ordering~$\tau$. 
By construction,~$\score(\widehat{A}) = t-k$. An example of the construction is shown in Figure~\ref{Figure: Example Reduction}

\begin{figure}
\begin{center}
\begin{tikzpicture}[scale=0.85,yscale=0.7]
\tikzstyle{knoten}=[circle,fill=white,draw=black,minimum size=6pt,inner sep=0pt]
\tikzstyle{bez}=[inner sep=0pt]

\node[knoten,label=below:$v_1$] (v1) at (-8,5) {};
\node[knoten,label=below:$v_2$] (v2) at (-6,4) {};
\node[knoten,label=below:$v_3$] (v3) at (-4,5) {};

\draw[-, line width=1pt]  (v1) to (v2);
\draw[-, line width=1pt]  (v3) to (v2);

\draw[rounded corners, fill=black!10, draw=black!10] (-11.5, -0.3) rectangle (-5.5, 0.7) {};
\node[bez] at (-13,0) {$\tau$};

\node[knoten,label=below:$w_1$ \vphantom{$w^1$}] (w1) at (-11,0) {};
\node[knoten,label=below:$w_1^1$] (w11) at (-10,0) {};
\node[knoten,label=below:$w_1^2$] (w12) at (-9,0) {};

\node[knoten,label=below:$w_2$ \vphantom{$w^1$}] (w2) at (-8,0) {};
\node[knoten,label=below:$w_2^1$] (w21) at (-7,0) {};
\node[knoten,label=below:$w_2^2$] (w22) at (-6,0) {};

\node[knoten,label=below:$x_1$ \vphantom{$w^1$}] (x1) at (-4.5,0) {};
\node[knoten,label=below:$x_2$ \vphantom{$w^1$}] (x2) at (-3.5,0) {};

\draw[rounded corners, fill=black!10, draw=black!10] (-2.5, -0.3) rectangle (0.5, 0.7) {};
\node[knoten,label=below:$v_1$ \vphantom{$w^1$}] (v1) at (-2,0) {};
\node[knoten,label=below:$v_2$ \vphantom{$w^1$}] (v2) at (-1,0) {};
\node[knoten,label=below:$v_3$ \vphantom{$w^1$}] (v3) at (0,0) {};

\draw[->, line width=1pt]  (w1) to (w11);
\draw[->, line width=1pt, bend left=25]  (w1) to (w12);

\draw[->, line width=1pt]  (w2) to (w21);
\draw[->, line width=1pt, bend left=25]  (w2) to (w22);

\draw[->, line width=1pt]  (x1) to (x2);

\draw[->, line width=1pt]  (x2) to (v1);
\draw[->, line width=1pt, bend left=20]  (x2) to (v2);
\draw[->, line width=1pt, bend left=22]  (x2) to (v3);

\begin{scope}[yshift=-4cm, xshift=2cm]
\draw[rounded corners, fill=black!10, draw=black!10] (-11.5, -0.3) rectangle (-5.5, 0.7) {};
\node[bez] at (-15,0) {$\tau'$};

\node[knoten,label=below:$w_1$ \vphantom{$w^1$}] (w1) at (-11,0) {};
\node[knoten,label=below:$w_1^1$] (w11) at (-10,0) {};
\node[knoten,label=below:$w_1^2$] (w12) at (-9,0) {};

\node[knoten,label=below:$w_2$ \vphantom{$w^1$}] (w2) at (-8,0) {};
\node[knoten,label=below:$w_2^1$] (w21) at (-7,0) {};
\node[knoten,label=below:$w_2^2$] (w22) at (-6,0) {};

\node[knoten,label=below:$x_1$ \vphantom{$w^1$}] (x1) at (-4.5,0) {};
\node[knoten,label=below:$x_2$ \vphantom{$w^1$}] (x2) at (-3.5,0) {};

\draw[rounded corners, fill=black!10, draw=black!10] (-2.5, -0.3) rectangle (-1.5, 0.7) {};
\node[knoten,label=below:$v_1$ \vphantom{$w^1$}] (v1) at (-2,0) {};
\node[knoten,label=below:$v_2$ \vphantom{$w^1$}] (v2) at (-12,0) {};
\node[knoten,label=below:$v_3$ \vphantom{$w^1$}] (v3) at (-13,0) {};

\draw[->, line width=1pt]  (w1) to (w11);
\draw[->, line width=1pt, bend left=25]  (w1) to (w12);

\draw[->, line width=1pt]  (w2) to (w21);
\draw[->, line width=1pt, bend left=25]  (w2) to (w22);

\draw[->, line width=1pt]  (x1) to (x2);

\draw[->, line width=1pt]  (x2) to (v1);

\draw[->,  line width=1pt, bend left=15]  (v3) to (w1);
\draw[->,  line width=1pt]  (v2) to (w1);

\draw[->, line width=1pt, bend left=25]  (v3) to (w2);
\draw[->, line width=1pt, bend left=20]  (v2) to (w2);
\end{scope}
\end{tikzpicture}
\end{center}
\caption{An example of the construction from the proof of Theorem~\ref{Theorem: W1-h}. The upper part shows the graph~$G$ of a clique instance~$(G,2)$. Below, there are the vertices of the~\IV instance~$I$, the ordering~$\tau$, and the arc set~$\widehat{A}$. The first six vertices on~$\tau$ correspond to the edge set of~$G$ and the last three vertices on~$\tau$ correspond to the vertex set of~$G$. The lower part shows a~$2$-close ordering~$\tau'$ together with an optimal arc set. The parent sets of~$w_1$ and~$w_2$ correspond to the clique on the vertices~$v_2$ and~$v_3$ in~$G$.}\label{Figure: Example Reduction}
\end{figure}
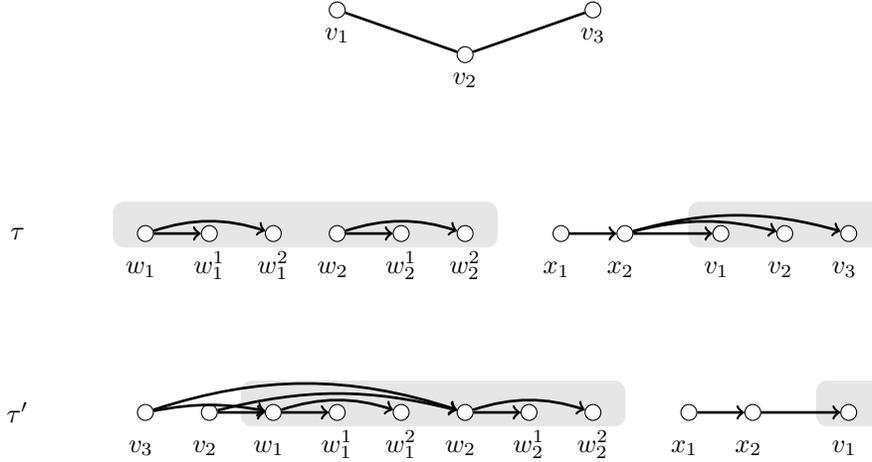

The idea is that every arc set~$A$ of score at least~$t$ has to contain all the arcs of~$A^*$. 
Moreover, if~$D=(N,A)$ has a topological ordering~$\tau'$ which is~$r$-close to~$\tau$, then at most~$r$ of the vertices of~$N$ change their position. 
Intuitively, vertices~$S\subseteq V$ should be inserted at the beginning of the new ordering~$\tau'$ so that for each edges~$e_i\in E(S)$, the vertex~$w_i$ can learn the parent set~$e_i$.
To obtain the maximal score, the number of edges in~$E(S)$ should be maximal, which is the case if and only if~$S$ is a clique of size~$k=r$ in~$G$. 

\iflong
Next, we show that~$I$ is a yes-instance of~\textsc{Clique} if and only if~$I'$ is a yes-instance of~\IV. Since~$k'=0$, we set~$\score(A):=\score(A,0)$ for ease of notation.

\proofWhardness
\fi

Next, we describe how we can modify this construction to obtain the hardness result for \SW.
The idea is that we add~$k$ additional vertices at the beginning of the topological ordering that need to be swapped with the~$k$ vertex of a clique~$S$ to obtain the required score.
We add~$k$ additional vertices~$y_1, \dots, y_k$ with~$\Fa_2(y_j)
:= \{(\emptyset,0,0),(\{x_k\},n^8,0)\}$ for each~$j\in [1,k]$ and~$\Fa_2(v) := \Fa(v)$ for each~$v\in N$.
Moreover, we set~$\sigma = (y_1, \dots, y_k) \cdot \tau$ and~$t_2 := t + k\cdot n^8$.
It remains to show that~$I'$ is a yes-instance of~\IV if and only if~$I_2 :=(N_2, \Fa_2, \sigma, t_2, k', r)$ is a yes-instance of \SW.
By construction and the above argumentation, an arc set~$A$ has score at least~$t_2-k = \score(\widehat{A}) + k\cdot n^8$ if and only if~$A^*\cup \{(x_k, y_j)\mid 1\leq j \leq k\}\subseteq A$. 
Consequently, each swap operation has to swap a different vertex from~$\{y_1, \dots, y_k\}$ with a vertex which appears after~$x_k$ in the current topological ordering, that is, with a vertex of~$V$. 
Let~$S$ be the vertices of~$V$ that are swapped with the vertices of~$\{y_1, \dots, y_k\}$. 
Then,~$\score(A) \geq t_2$ if and only if~$S$ forms a clique in~$G$.
\end{proof}

\section{Parameterized Local Search for Inversions Distance}
In this section we study parameterized local search for the \emph{inversions} distance. An inversion on an ordering is an operation that swaps the positions of two consecutive vertices of the ordering. 

We describe a randomized algorithm to solve~\MWBNSL. The algorithm has constant error probability and runs in subexponential FPT time for~$r$, the radius of the local search neighborhood. The algorithm is closely related to a parameterized local search algorithm for finding minimum weight feedback arc sets in tournaments by Fomin et al.~\cite{FLRS10}. A \emph{tournament} is a directed graph~$T=(N,A)$ where either $(u,v) \in A$ or~$(v,u) \in A$ for every pair of distinct vertices~$u$ and~$v$. A \emph{feedback arc set} is a subset~$A' \subseteq A$ such that~$(N,A \setminus A')$ is a DAG. The problem is defined as follows.

\begin{center}
	\begin{minipage}[c]{.9\linewidth}
          \textsc{Feedback Arc Set in Tournaments (FAST)}\\
          \textbf{Input}: 
A tournamend~$T:=(N,A)$, a weight function~$\omega: A \rightarrow \mathds{N}$, and an integer~$k$.\\
          \textbf{Question}: Is there a feedback arc set~$A' \subseteq A$ such that~$\sum_{a \in A'} \omega(a) \leq k$?
	\end{minipage}
\end{center}

In a local search version of FAST one is given a feedback arc set~$A'$ and aims to find a feedback arc set~$A''$ where the sum of the weights is strictly smaller,~$|A' \setminus A''| \leq r$, and~$|A'' \setminus A'| \leq r$ for a given search radius~$ r \in \mathds{N}_0$. The approach of Fomin et al. \cite{FLRS10} is ordering-based in the following sense: A feedback arc set~$A'$ is identified with the unique topological ordering~$\tau$ of~$(N,A_\text{rev})$, where~$A_{\text{rev}}$ results from~$A$ by reversing all directions of arcs in~$A'$. Interchanging the positions of two consecutive vertices on~$\tau$ then corresponds to removing exactly one arc from the current feedback arc set (or adding it to the feedback arc set, respectively).

FAST and \WBNSL are related, since in both problems we aim to find a DAG. Furthermore, removing one arc~$(u,v)$ in a tournament can be seen as changing the parent set of~$v$. However, note that the seemingly more general problem~\WBNSL is not an actual generalization of FAST: A parent set of a vertex contains up to~$n-1$ vertices and therefore, to model all posible deletions of its incoming arcs with local multiscores, there are up to~$2^{n-1}$ potential parent sets. Thus, modelling an instance of FAST as an instance of \WBNSL in this natural way takes exponential time and space. For our algorithm, we adapt the techniques used for~\textsc{FAST} and show that we can use these to obtain a local search algorithm for~\textsc{W-BNSL}.

We first introduce some formalism of inversions. Let~$\tau=(v_1, \dots, v_n)$ be an ordering of a set~$N$. An \emph{inversion on position~$i \in [1, n-1]$} transforms~$\tau$ into the ordering~$h_i(\tau) := (v_1, \dots, v_{i-1},v_{i+1},v_i, \dots, v_n)$. A \emph{sequence of inversions}~$S=(s_1, \dots, s_\ell)$ is a finite sequence of elements in~$[1, n-1]$. Applying~$S$ on~$\tau$ transforms~$\tau$ into the ordering~$S(\tau):=h_{s_\ell}\circ h_{s_{\ell-1}} \circ \dots \circ h_{s_1}(\tau)$. The distance~$\Inv(\tau,\tau')$ is the length of the shortest sequence of inversions~$S$ such that~$S(\tau)=\tau'$. Recall that the inversions-distance is also known as the \emph{Kendall tau distance}~\cite{K38}.

The key idea behind the algorithm is as follows: If we know a partition of the vertex set~$N$ into relatively few sets which do not change their relative positions in the topological ordering, then we have a limited space of possible orderings which can be obtained by performing at most~$r$ inversions. To specify which vertices keep their relative positions we employ a technique called color coding~\cite{AYZ95}. Intuitively, we randomly color all vertices with~$\Oh(\sqrt{r})$ colors in a way that vertices of the same color keep their relative positions. As the local search algorithm for FAST~\cite{FLRS10}, our color coding algorithm is closely related to a subexponential time algorithm for \textsc{FAST}~\cite{ALS09}. 

To describe the color coding technique, we need some definitions: Let~$N$ be a set of vertices. A function~$\chi: N \rightarrow [1,\ell]$ is called a~\emph{coloring (of~$N$ with~$\ell$ colors)}. For each color~$i \in [1,\ell]$, we call~$Z^i:=\{v \in N \mid \chi(v)=i\}$ the~\emph{color class of~$i$}. We next define color-restricted arc sets and color-restricted solutions which are important for the color coding algorithm.

\begin{definition}
Let~$I:=(N,\Fa,\tau,k,r)$ be an instance of~\MWBNSL , let~$\chi:N \rightarrow [1, \ell]$ be a coloring, and let~$A$ be an~$\Fa$-valid arc set. We say that~$A$ is a \emph{color-restricted arc set} if there is a topological ordering~$\tau'$ of~$(N,A)$ such that~$\Inv(\tau,\tau') \leq r$ and~$\tau[Z^i]=\tau'[Z^i]$ for every color~class~$Z^i$. 
A color-restricted arc set~$A$ that maximizes~$\score(A,k)$ is called a \emph{color-restricted solution} of~$I$ and~$\chi$.
\end{definition}

We next describe a deterministic algorithm that efficiently finds a color-restricted solution.

\begin{proposition} \label{Prop: Solving Basement Problem}
Given an instance~$I:=(N,\Fa,\tau,k,r)$ of~\MWBNSL  and a coloring~$\chi:N \rightarrow [1, \ell]$ for~$I$, a color-restricted solution~$A$ can be computed in~$(r+2)^\ell \cdot \poly(|I|)$~time.
\end{proposition}

\begin{proof}
Before we present the algorithm, we provide some intuition.

\textit{Intuition.} Our algorithm is based on dynamic programming. Every color class can be seen as a chain of vertices that keep their relative position in the ordering. An integer vector of length~$\ell$ then describes which prefixes of these chains we consider. Our dynamic programming algorithm starts with empty chains and then adds the next vertex of one of the chains and finds a solution of the instance with the extended prefix vectors.

\textit{Notation.}  To formally describe the algorithm we introduce some notation. 
For every color class~$Z^i$ consider the sub-ordering~$\tau[Z^i]$ that contains only the vertices of~$Z^i$. Given some integer~$x \leq |Z^i|$ we define~$Z^i_x:= \{z_i(1), z_i(2), \dots, z_i(x)\}$ as the set of the first~$x$ vertices appearing on~$\tau[Z^i]$, where~$z_i(j)$ denotes the~$j$th vertex on~$\tau[Z^i]$.

Note that~$Z^i_0= \emptyset$. Given an integer vector~$\vec{p}=(p_1, \dots, p_\ell)$ with~$p_i \in [0, |Z^i|]$, we define~$\tau(\vec{p}):= \tau[Z^1_{p_1} \cup Z^2_{p_2} \cup \dots \cup Z^\ell_{p_\ell}]$. 
As a shorthand, for the set of vertices appearing on~$\tau(\vec{p})$ we define~$N(\vec{p}):=N(\tau(\vec{p}))$ and we define~$\Fa_{\vec{p}}$ by~$\Fa_{\vec{p}}(v):= \{(P,s,\w) \in \Fa(v) \mid P \subseteq N(\vec{p}) \}$ for every~$v \in N(\vec{p})$ as the \emph{limitation of~$\Fa$ to~$N(\vec{p})$}.  
Note that, given a partial ordering~$\tau(\vec{p})$, the vertex~$z_i(p_i)$ is the last vertex of color class~$Z^i$ appearing on~$\tau(\vec{p})$. 
Throughout this proof, we let~$\vec{0}$ denote the integer vector of length~$\ell$ where all entries equal zero, and~$\vec{e}_i$ the integer vector of length~$\ell$ where the~$i$th entry equals one and all other entries equal~zero.

\textit{Algorithm.} The dynamic programming table~$T$ has entries of the type~$T[\vec{p},k',r']$ with~$k' \in [0, k]$ and~$r' \in [0, r]$. Each entry stores the score of a color-restricted solution of the~\MWBNSL instance
\begin{align*}
I^{\vec{p}}_{k',r'} := ( N(\vec{p}), \Fa_{\vec{p}}, \tau(\vec{p}), k',r').
\end{align*}
The instance~$I^{\vec{p}}_{k',r'}$ is an instance obtained by only considering the prefixes of the chains of color classes which are specified by the vector~$\vec{p}$. The idea behind this algorithm is to recursively find the best sink of the current network and combine this with a color-restricted solution of the remaining network. To specify the contribution of a sink to the score, we introduce the following definition: For given~$i \in [0, \ell]$, $k' \in [0, k]$, and~$\vec{p}$, we define the value~$f_{\vec{p}}(i,k')$ as the maximal local score of a parent set~$P \subseteq N(\vec{p})$ of~$z_i(p_i)$ that simultaneously has weight at most~$k'$. More formally,
\begin{align*}
f_{\vec{p}}(i,k') := \max_{\substack{(P,s,\w) \in \Fa(z_i(p_i))\\ P \subseteq N(\vec{p}) \\ \w \leq k'}} s.
\end{align*}
The value of~$f_{\vec{p}}(i,k')$ can be computed in~$\poly(|I|)$ time by iterating over the array representing~$\Fa(v)$.

We next describe how to fill the dynamic programming table. As base case we set~$T[\vec{0},k',r'] := 0$ for all~$k' \in [0, k]$ and~$r' \in [0, r]$. 
The recurrence to compute entries for~$\vec{p} \neq \vec{0}$ is
\begin{align*}
T[\vec{p},k',r'] := & \max_{k'' \leq k'} ~ \max_{\substack{i,~ p_i > 0\\R(\vec{p},i) \leq r'}} 
\Bigl( f_{\vec{p}}(i, k'') + T[\vec{p}-\vec{e}_i, k'-k'', r'-R(\vec{p},i)]\Bigr),
\end{align*}
where~$R(\vec{p},i):= | \{ v \in N(\vec{p}) \mid z_i({p_i}) <_{\tau(\vec{p})} v\}|$ is the number of elements that appear after~$z_i({p_i})$ in~$\tau(\vec{p})$. Intuitively,~$R(\vec{p},i)$ is the number of inversions that need to be performed to move~$z_i(p_i)$ to the end of the ordering~$\tau(\vec{p})$. The score of a color-restricted solution of~$I$ and~$\chi$ can be computed by evaluating~$T[(|Z^1|,|Z^2|, \dots, |Z^{\ell}|),k,r]$. The corresponding network can be found via traceback.

\textit{Correctness.} We next show that the dynamic programming recurrence is correct. That is, we prove the following claim.

\begin{claim} \label{Claim: Correctness inner DP}
The table entry~$T[\vec{p},k',r']$ contains the score of a color-restricted solution of~$I^{\vec{p}}_{k',r'}$ and~$\chi|_{N(\vec{p})}$ for each combination~$\vec{p}$, $k'$, and~$r'$.
\end{claim}

\iflong
\begin{claimproof}
We prove the claim by induction over~$|\vec{p}|:= \sum_{i=1}^{\ell} p_i$.

\textit{Base Case: $|\vec{p}|=0$.} Then~$\vec{p}= \vec{0}$ and~$I^{\vec{0}}_{k',r'}$ is an instance with an empty vertex set. Thus, the score of a color-restricted solution is~$0 = T[\vec{0},k',r']$.

\textit{Inductive Step:} Let the claim hold for all~$\vec{q}$ with~$|\vec{q}| < |\vec{p}|$. Let~$A$ be a color-restricted solution of~$I^{\vec{p}}_{k',r'}$ and~$\chi|_{N(\vec{p})}$. We prove
\begin{align*}
\score(A, k') = T[\vec{p},k',r'].
\end{align*}

$(\geq)$ We first show~$\score(A,k') \geq T[\vec{p},r',k']$. Let~$i$ and $k''$ be the integers that maximize the right hand side of the recurrence. Taking the vertex~$z_i({p_i})$ and moving it to the rightmost position of~$\tau(\vec{p})$ can be done with exactly~$R(\vec{p},i) \leq r'$ inversions. Let~$P$ be the parent set of the tuple~$(P,s,\omega)$, that maximizes~$f_{\vec{p}}(i,k'')$. Defining~$P$ as the parent set with weight at most~$k''$ of~$z_i(p_i)$ and combining this with a color-restricted solution~$A''$ of~$I^{\vec{p}-\vec{e}_i}_{k'-k'',r'-R(\vec{p},i)}$ defines a color-restricted arc set~$A'$ for~$I^{\vec{p}}_{k',r'}$ with~$\score(A',k') = \score(A'',k'-k'') + f_{\vec{p}}(i, k'')$. The inductive hypothesis implies~$\score(A'',k'-k'')=T[\vec{p}-\vec{e}_i,k'-k'',r'-R(\vec{p},i)]$. We thus have~$\score(A',k') = T[\vec{p},k',r']$, and therefore~$\score(A,k') \geq T[\vec{p},k',r']$.

$(\leq)$ We next show~$\score(A,k') \leq T[\vec{p},r',k']$. Let~$Y^1, \dots, Y^{\ell}$ be the color classes of~$\chi|_{N(\vec{p})}$. Since~$A$ is a color-restricted solution of~$I^{\vec{p}}_{k',r'}$ and~$\chi|_{N(\vec{p})}$, there exists a topological
 ordering~$\tau'$ of~$(N(\vec{p}),A)$ with~$\Inv(\tau(\vec{p}), \tau') \leq r'$ and~$\tau(\vec{p})[Z^i]=\tau'[Y^i]$ for every color~$i$. It follows that the last vertex appearing on~$\tau'$ is~$z_i({p_i})$ for some color~$i$. By the inductive hypothesis the entry~$T[\tau(\vec{p}-\vec{e}_i), k'-k'', r'-R(\vec{p},i)]$ stores the score of a color-restricted solution of~$I^{\vec{p}-\vec{e}_i}_{k'-k'',r'-R(\vec{p},i)}$ for each~$i$ and~$k'$. Since we iterate over every choice of~$i$ and~$k''$, we consider every possible choice of~$z_i({p_i})$ and combine its best possible parent set with a color-restricted solution of the remaining instance. Consequently,~$\score(A,k') \leq T[\vec{p},k',r']$.  $\hfill \Diamond$
\end{claimproof}
\fi

\textit{Running Time.} We now analyze the running time of the algorithm. The table~$T$ has~$\Oh(n^\ell \cdot (k+1) \cdot (r+1))$ entries. Each entry can be computed in time polynomial in~$|I|$. This would lead to a running time of~$n^{\ell} \cdot \poly(|I|)$. However, we can obtain a running time of~$(r+1)^\ell \cdot \poly(|I|)$ by using a more elaborate analysis similar to the one used by Fomin et al.~\cite{FLRS10}. The key idea is to compute the entry~$T[(|Z^1|,|Z^2|, \dots, |Z^{\ell}|),k,r]$ in a top-down manner using a memoization table to store the results of computed table entries. To obtain the running time bound, we give a bound on the number of vectors~$\vec{p}$ such that an entry~$T[\vec{p},k',r']$ is evaluated by the algorithm in order to compute the entry~$T[(|Z^1|,|Z^2|, \dots, |Z^{\ell}|),k,r]$. Throughout the rest of this proof we call such~$\vec{p}$ an \emph{important vector}.

Let~$\vec{p}$ be an important vector. An index~$j \in [1,n]$ is called a~\emph{hole in~$\tau(\vec{p})$} if~$\tau(j) \not \in N(\vec{p})$ and there exists some~$j'>j$ with~$\tau(j') \in N(\vec{p})$. Given a hole~$j$, we let~$\Phi(j)$ denote the set of vertices in~$N(\vec{p})$ that appear after~$\tau(j)$ on~$\tau(\vec{p})$. 

The intuition behind holes is the following: Throughout the top-down algorithm we have a `budget' of at most~$r$ inversions. A hole at some position~$j$ is caused by a recursive call of the algorithm with some vector~$\vec{p}-\vec{e}_i$ where~$i$ corresponds to the color of the vertex~$\tau(j)$. Intuitively, this means that the vertex~$\tau(j)$ is chosen to be a sink in a possible solution of an instance~$I^{\vec{p}}_{k',r'}$ and therefore it is moved to the last position of the ordering~$\tau(\vec{p})$. For this movement, the vertex~$\tau(j)$ needs to pass all the vertices appearing on~$\tau(\vec{p})$ after~$\tau(j)$. Consequently, this decreases our budget of inversions by at least~$\Phi(j)$. We then use the fact that the total budget of inversions is at most~$r$ which implies that the number and the position of holes in~$\tau(\vec{p})$ is restricted when~$\vec{p}$ is an important vector. Afterwards, we use this restriction to give an upper bound on the total number of important vectors. To formally show that the number of important vectors is bounded, we need the inequality stated in the following claim.

\begin{claim} \label{Claim: Hole Inequality}
If~$T[\vec{p},k',r']$ is an entry that is evaluated by the algorithm, we have
\begin{align*}
\sum_{j {\rm ~is~a~hole~in~}\tau(\vec{p})} |\Phi(j)| \leq r.
\end{align*}
\end{claim}

\begin{claimproof}
Consider a sequence of table entries~$(T[{\vec{q} \, }^1,k_1,r_1] , \dots, T[{\vec{q} \, }^s,k_s,r_s])$ such that
\begin{align*}
T[{\vec{q} \, }^1,k_1,r_1] &= T[(|Z^1|,|Z^2|, \dots, |Z^{\ell}|),k,r],\\
T[{\vec{q} \, }^s,k_s,r_s] &= T[\vec{p}, k',r'],
\end{align*}
and for each~$x \in [1, s-1]$ the entry~$T[{\vec{q} \, }^{x+1},k_{x+1},r_{x+1}]$ has been evaluated by the algorithm in order to compute the entry~$T[{\vec{q} \, }^x,k_x,r_x]$. Such sequence exists, since~$T[\vec{p},k',r']$ is evaluated by the algorithm. Note that~$\sum_{x=1}^{s-1} (r_x - r_{x+1}) \leq r$.

Let~$j$ be a hole of~$\tau(\vec{p})$. Then, there is a unique index~$x \in [1, s-1]$ such that the computation of~$T[{\vec{q} \, }^x,k_x,r_x]$ creates the hole~$j$ by recursively calling~$T[{\vec{q} \, }^{x+1},k_{x+1},r_{x+1}]$ with~${\vec{q} \, }^{x+1}={\vec{q} \,}^x-e_i$ and~$z_i({q^x}_i)=\tau(j)$ for some~$i$. By the recurrence we then have~$r_x = R({\vec{q} \, }^x,i)+ r_{x+1}$. Therefore,~$r_x-r_{x+1} = R({\vec{q} \, }^x,i) \geq |\Phi(j)|$ since~$\tau({\vec{p}})$ is an induced subordering of~$\tau({\vec{q} \, }^x)$. Thus, we have
\begin{align*}
\sum_{j \text{ is a hole in }\tau(\vec{p})} |\Phi(j)| \leq \sum_{x=1}^{s-1} (r_x - r_{x+1}) \leq r.
\end{align*} $\hfill \Diamond$
\end{claimproof}

We next use Claim~\ref{Claim: Hole Inequality} to give a bound on the number of important vectors. Let~$\vec{p}$ be important and let~$x \in [1,n]$ such that~$\tau(x)$ is the last vertex appearing on~$\tau(\vec{p})$. We show that~$\tau(j) \in N(\vec{p})$ for every~$j \in [1, x-(r+1)-1]$. Intuitively, this means that all holes in~$\tau(\vec{p})$ are close to~$x$. Assume towards a contradiction that there is some~$j \in [1, x-(r+1)-1]$ with~$\tau(j) \not \in N(\vec{p})$. Note that~$j$ is a hole in~$\tau(\vec{p})$ with~$\tau(x) \in \Phi(j)$. Consider an index~$j' \in [x-(r+1),x-1]$. If~$\tau(j') \not \in N(\vec{p})$, then~$j'$ is a hole in~$\tau(j')$ with~$\tau(x) \in \Phi(j')$. Otherwise, if~$\tau(j') \in N(\vec{p})$, we have~$\tau(j') \in \Phi(j)$. Since there are~$r+1$ possible values for~$j'$, this is a contradiction to the inequality from Claim~\ref{Claim: Hole Inequality}.

We can specify an important vector~$\vec{p}$ by specifying the set~$N(\vec{p})$. By the above, it suffices to specify which elements of~$X:=\{\tau(x-(r+1)), \dots, \tau(x-1)\}$ belong to~$N(\vec{p})$. Recall that there are~$\ell$ color classes, and the vertices of each color class~$Z^i$ which belong to~$N(\vec{p})$ are defined by the entry~$p_i$. Since~$X$ contains~$r+1$ elements, for each color class, there are~$r+2$ possible choices of~$Z^i \cap X$. Thus, there are~$(r+2)^\ell$ possible important vectors. Consequently, the algorithm computes at most~$(r+2)^\ell \cdot (k+1) \cdot (r+1)$ table entries. Since each entry can be computed in polynomial time in~$|I|$, the algorithm runs in~$(r+2)^\ell \cdot \poly(|I|)$~time.
\end{proof}

We now describe how to use the algorithm behind Proposition~\ref{Prop: Solving Basement Problem} to obtain a randomized algorithm for~\MWBNSL. We randomly color the vertices with~$\sqrt{8r}$ colors and use the algorithm from Proposition~\ref{Prop: Solving Basement Problem} to find a color-restricted solution. For ease of notation, we assume that~$\sqrt{8r}$ is an integer. Given an instance~$I$ and a coloring~$\chi$, we say that~$\chi$ is a \emph{good coloring of~$I$} if every color-restricted solution of~$\chi$ and~$I$ is a solution of~$I$. Note that not every coloring is a good coloring. A trivial example is a random coloring, where every vertex receives the same color. In this case, a colored solution has the topological ordering~$\tau$ since all vertices keep their relative positions. However, if~$r \geq 1$, the uncolored instance might have strictly better solution with a topological ordering that is~$r$-close to~$\tau$.

 We next analyze the likelihood of randomly choosing a good~coloring.

\begin{lemma}\label{Lemma: successful coloring}
Let~$I$ be an instance of \MWBNSL  and let~$\chi:N \rightarrow [1, \sqrt{8r}]$ be a coloring that results from assigning a color to each vertex uniformly at random. Then,~$\chi$ is a good coloring of~$I$ with probability at least~$(2e)^{-\sqrt{\nicefrac{r}{8}}}$.
\end{lemma}

\iflong
\begin{proof}
Let~$I=(N,\Fa,\tau,k,r)$, let~$A$ be a solution of~$I$, and let~$\tau'$ be an~$r$-close topological ordering of~$(N,A)$. Consider the graph~$G:=(N,E)$, where~$E:=\{ (u,v) \mid u <_{\tau} v \text{ and } v <_{\tau'} u \}$. Intuitively, an arc~$(u,v) \in E$ indicates that in~$\tau'$ the vertices~$u$ and~$v$ have another relative position than in~$\tau$. Note that, when applying a sequence of inversions~$S$ on an orderingtransforming an ordering, at most~$|S|$ pairs of vertices change their relative positions. Thus, we have~$|E| \leq r$ since~$\tau'$ is~$r$-close to~$\tau$.

A proper vertex coloring of~$G$ is a coloring that assigns distinct colors to every pair of vertices that are connected by an edge. It is easy to see that~$\chi$ is a good coloring of~$I$ if~$\chi$ is a proper vertex coloring of~$G$. The probability of randomly choosing a proper vertex coloring with~$\sqrt{8r}$ colors for a graph with~$r$ edges is at least~$(2e)^{-\sqrt{\nicefrac{r}{8}}}$~\cite{ALS09}.
\end{proof}
\fi

\iflong We next describe the randomized algorithm. \fi
Recall that for every vertex~$v$ we have~$(\emptyset,s,0) \in \Fa(v)$ for some~$s \in \mathds{N}_0$. Thus, given an instance~$(N,\Fa,\tau,k,r)$ of~\MWBNSL, there always exists an~$\Fa$-valid arc set~$A:=\emptyset$ with weight at most~$k$ such that~$(N,A)$ has a topological ordering that is~$r$-close to~$\tau$.

\begin{theorem} \label{Theorem: Inversions Algorithm}
There exists a randomized algorithm for~\MWBNSL  that, in time~$2^{\Oh(\sqrt{r} \cdot \log(r))} \cdot \poly(|I|)$ returns an~$\Fa$-valid arc set~$A$ with weight at most~$k$ such that~$(N,A)$ has a topological ordering that is~$r$-close to~$\tau$. With probability at least~$1- \frac{1}{e}$, the arc set~$A$ is a solution.
\end{theorem}

\iflong
\begin{proof}
Let~$I$ be an instance of the~\MWBNSL. The algorithm can be described as follows: Repeat the following two steps~$(2e)^{\sqrt{\nicefrac{8}{r}}}$~times and return the color-restricted solution with the maximal~score.
\begin{enumerate}
\item[1.] Color every vertex of~$N$ independently with one color from the set~$[1, \sqrt{8r}]$ with uniform probability. Let~$\chi$ be the resulting coloring.
\item[2.] Apply the algorithm from Proposition~\ref{Prop: Solving Basement Problem} and compute a color-restricted solution of~$I$ and~$\chi$.
\end{enumerate}

By Lemma~\ref{Lemma: successful coloring}, the probability of choosing a good coloring in Step 1 and therefore computing a solution of~$I$ in Step 2 is at least~$(2e)^{-\sqrt{\nicefrac{8}{r}}}$. Thus, by repeating these steps~$(2e)^{\sqrt{\nicefrac{8}{r}}}$ times, we obtain a running time of~$(2e)^{\sqrt{\nicefrac{r}{8}}} \cdot (r+2)^{\sqrt{8r}} \cdot \poly(|I|)$. As a shorthand, let~$x:= \sqrt{\nicefrac{8}{r}}$. The probability that the output is not a solution is at most
\begin{align*}
(1-(2e)^{-x})^{(2e)^{x}} \leq (e^{-(2e)^{-x}})^{(2e)^{x}} = \frac{1}{e}.
\end{align*}
The first inequality relies on the inequality~$(1+y) \leq e^y$ for all~$y$. Then, the probability that the output is a solution is at least~$1-\frac{1}{e}$.
\end{proof}
\fi

\section{Parameterized Local Search for Inversion-Window Distance} \label{Subsection: Combine Solutions}

We now study the \emph{inversion-window distance}~$\text{InvWin}(\tau,\tau')$ which is closely related to the number of inversions. The intuition behind the new distance is the following: Given an ordering~$\tau$, we obtain an~$r$-close ordering by partitioning~$\tau$ into multiple windows of consecutive vertices and performing up to~$r$~inversions inside each such window. This new distance is an extension of the inversions distance in the sense that~$\text{InvWin}(\tau,\tau') \leq \text{Inv}(\tau,\tau')$. In other words, given a search radius~$r$, the search space of~$r$-close orderings is potentially larger when using the~$\text{InvWin}(\tau,\tau')$ instead of~$\text{Inv}(\tau,\tau')$. For this extended distance we also provide a randomized algorithm with running time~$2^{\Oh(\sqrt{r} \log(r))} \cdot \poly(|I|)$.

We now formally define the inversion-window distance. Let~$\tau$ and~$\tau'$ be orderings of a vertex set~$N$ and let~$S$ be a sequence of inversions that transforms~$\tau$ into~$\tau'$. A \emph{window partition}~$W$ for~$S$ is a partition of the set~$[1,n]$ into intevals~$I_1 = [a_1,b_1], \dots,I_\ell = [a_\ell,b_\ell]$ such that no endpoint~$b_i$ appears on~$S$. Note that for each window~$I_j=[a_j,b_j]$ we have~$N_{\tau}(a_j,b_j)=N_{\tau'}(a_j,b_j)$. In other words, given a window partition~$W$, applying~$S$ on~$\tau$ does not move any vertex from one interval of~$W$ to another interval of~$W$. A window partition always exists since~$n$ does not appear on any sequence of inversions and thus~$W:=\{ [1,n]\}$ is a window partition. The width of a window partition~$W$ for~$S$ is defined as

\begin{align*}
\text{width}(W):= \max_{I \in W} \sum_{j \in I} \#(S,j),
\end{align*}

where~$\#(S,j)$ denotes the number of appearances of index~$j$ on the sequence~$S$. The \emph{number of window inversions}~$\WI(S)$ is then defined as the minimum~$\text{width}(W)$ among all window partitions~$W$ of~$S$. In other words, $\WI(S)$ is the smalles possible maximum number of inversions inside a window among all possible window partitions. The \emph{inversion-window distance}~$\text{InvWin}(\tau,\tau')$ is the minimum number~$\WI(S)$ among all sequences~$S$ that transform~$\tau$ into~$\tau'$. Observe that~$\text{InvWin}(\tau,\tau') \leq \text{Inv}(\tau,\tau')$ since~$\{[1,n]\}$ is a window partition of every sequence~$S$.

Our algorithm is based on the following intuition: If an ordering~$\tau'$ is~$r$-close to~$\tau$, the ordering can be decomposed into windows in which at most~$r$ inversions are performed and the vertices from distinct windows keep their relative positions. In a bottom-up manner, we compute the best possible ordering of suffixes~$\tau(a,n)$ of~$\tau$ until we find the best possible ordering of~$\tau(1,n)=\tau$. This is done by finding  the first window on the suffix and combining this with an optimal ordering of the remaining vertices of this suffix. To find a solution of the first window we use the algorithm behind Theorem~\ref{Theorem: Inversions Algorithm} for \MWBNSL as a subroutine.

The sub-instances on which we apply the algorithm behind Theorem~\ref{Theorem: Inversions Algorithm} contain the vertices of a \emph{window}~$\tau(a,b)$. Changing the ordering in~$\tau(a,b)$, these vertices may learn a parent set containing vertices from~$\tau(1,b)$. For our purpose, only the new parents in~$\tau(a,b)$ are important, since the new parents in~$\tau(1,a-1)$ are not important to find an optimal ordering of the suffix~$\tau(a,n)$. Intuitively, we hide the parents in~$\tau(1,a-1)$. Formally, we define the restricted local multiscores~$\Fa|_a^b$ by
$$\Fa|_a^b:=\{(P \cap N_\tau(a,b),s, \w) \mid (P,s,\w) \in \Fa(v) {\rm~and~} P \subseteq N_\tau(1,b) \}.$$

\begin{theorem} \label{Theorem: Window Inversions FPT Algorithm}
There exists a randomized algorithm for \wMWBNSL that, in time~$2^{\Oh(\log(r)\sqrt{r})}\cdot \poly(|I|)$, returns an~$\Fa$-valid arc set with weight at most~$k$. With probability at least~$1- \frac{1}{e}$, the returned arc set is a solution.
\end{theorem}

\begin{proof}
Let~$I:=(N,\Fa,\tau,k,r)$ be an instance of~\wMWBNSL. We describe the algorithm in two steps. We first describe a deterministic algorithm that finds a solution of~$I$ in polynomial time when using an oracle that gives solutions of \MWBNSL-instances where the search radius is at most~$r$. Afterwards, we describe how to replace the oracle evaluations with the randomized algorithm behind Theorem~\ref{Theorem: Inversions Algorithm} to obtain a randomized algorithm for~\wMWBNSL with the claimed running time and error probability.

\textit{The oracle algorithm.} We fill a dynamic programming table~$T$ that has entries of the type~$T[j,k']$ with~$j \in [1, n+1]$ and~$k' \in [0, k]$. The idea is that each entry~$T[j,k']$ with~$j \leq n$ stores the score of a solution of the \wMWBNSL-instance
\begin{align*}
I(j,k'):= (N_\tau(j,n), \Fa|_j^n, \tau(j,n), k',r).
\end{align*}

Intuitively, a solution of~$I(j,k')$ gives the best possible ordering of the suffix~$\tau(j,n)$ corresponding to an arc set of maximum weight~$k'$. We next describe how to fill~$T$. As base case we set~$T[n+1,k'] = 0$  for all~$k'$.
The recurrence to compute an entry for~$j \leq n$ is
\begin{align*}
T[j,k'] := \max_{p \in [j, n]} \max_{k'' \leq k'} \Bigl(W(j,p,k'') +T[p+1,k'-k'']\Bigr),
\end{align*}
where~$W(j,p,k'')$ denotes the score of a solution of the \MWBNSL-instance
\begin{align*}
(N_\tau(j,p), \Fa|^p_j, \tau(j,p), k'',r)
\end{align*}
which we obtain by an oracle evaluation. Intuitively, this instance corresponds to the first window of the suffix~$\tau(j,n)$. Observe that we evaluate the oracle on \MWBNSL-instances with search radius~$r$. The score of a solution of~$I$ can be computed by evaluating~$T[1,k]$. The ordering of the corresponding network can be found via traceback.
The correctness follows from the fact  we consider every possible position of the endpoint of the first window on the suffix~$\tau(j,n)$.

\textit{Replacing the Oracle Evaluations.} 
The dynamic programming table has~$(n+1)(k+1)$ entries. Each entry can be computed by using at most~$n \cdot k$ oracle evaluations. Thus, there are at most~$x := (n+1)(k+1) \cdot n \cdot k \in \poly(|I|)$ oracle evaluations.
We replace every oracle evaluation by applying the algorithm behind Theorem~\ref{Theorem: Inversions Algorithm} exactly~$x$ times and keeping a result with maximal score. 

Observe that the algorithm always computes a feasible arc set for~$I$. The probability that all~$x$ repetitions fail to compute a solution of the corresponding \MWBNSL-instance is at most~$(\frac{1}{e})^x$. Consequently, the probability that the correct result of one oracle evaluation is returned is at least~$1- \frac{1}{e^x}$. Therefore, the success probability of the algorithm is at least~$(1-\frac{1}{e^x})^x \geq (1- \frac{1}{e})$. The inequality holds since we have equality in the case of~$x=1$ and the left hand side of the inequality strictly increases when~$x \geq 1$ increases.

We next analyze the running time of the algorithm. As mentioned above, the dynamic programming table has~$(n+1)(k+1) \in \poly(|I|)$ entries. For each entry we apply the algorithm from Theorem~\ref{Theorem: Inversions Algorithm} on~$x$ instances of~\MWBNSL~with search radius~$r$. Since~$x \in \poly(|I|)$ we have a total running time of~$2^{\Oh(\log(r)\sqrt{r})}\cdot \poly(|I|)$.
\end{proof}

\section{Limits of Ordering-Based Local Search} \label{Section: Limits}

\iflong In this section we outline the limits of our framework. \fi As mentioned above, \WBNSL can be used to model Bayesian network learning under additional sparsity constraints like a bounded number of arcs. However, some important sparsity constraints are unlikely to be modeled with our framework, for example sparsity constraints that are posed on the moralized graph~\cite{EG08}:
 The \emph{moralized graph} of a DAG~$D$ is an undirected graph~$\Mo(D):=(V, E_1 \cup E_2)$ with~$V:=N$, $E_1 := \{\{u,v\} \mid (u,v) \in A \}$, and~$E_2 := \{\{u,v\} \mid u \text{ and }v \text{ have a common child in }D \}$. The edges in~$E_2$ are called~\emph{moral edges}. The NP-hard task of Bayesian inference can be solved more efficiently if the moralized graph of the network is treelike, that is, it has small treewidth~\cite{D09}. Thus, it is well motivated to learn a Bayesian network that satisfies a sparsity constraint that provides an upper bound on the treewidth of the moralized~graph.

\begin{table}
\begin{center}
  \begin{tabular}[t]{ll}
    \hline 
Version & Constraint on~$\Mo(D)$\\
\hline
\textsc{B-Edges-BNSL}~\cite{GK20} & bounded number of edges\\
\textsc{B-VC-BNSL}~\cite{KP15} & bounded vertex cover size \\
\textsc{B-DN-BNSL}~\cite{GK20} & bounded dissociation number\\
    \textsc{B-TW-BNSL}~\cite{KP13} & bounded treewidth\\
                         \hline
\end{tabular}
 \end{center}
\caption{\fontsize{10pt}{10pt} \selectfont Further sparsity constraints.}
 \label{Table: Moral Sparsity Constraints}
\end{table}

Table~\ref{Table: Moral Sparsity Constraints} displays versions of BNSL under sparsity constraints for the moralized graph.  We now argue that there is little hope that one can efficiently find improvements of a given DAG by applying changes on its topological ordering. For all problem versions in Table~\ref{Table: Moral Sparsity Constraints} it is NP-hard to learn a DAG of a given score~$t$ even when~$S_{\Fa}$ is a DAG and~$t$ is polynomial in~$n$ and all local scores are integers~\cite{KP13,KP15,GK20}. Observe that an acyclic superstructure has a topological ordering~$\tau$ which is then automatically a topological ordering of every DAG where the parent sets are potential parent sets.  Furthermore, observe that all the constraints on~$\Mo(D)$ from Table~\ref{Table: Moral Sparsity Constraints} are true if~$D$ does not contain arcs. Assume one can find improvements of a given DAG in polynomial time if the radius~$r$ of the local search neighborhood is constant. Then, we can solve instances with acyclic superstructure by setting~$r=0$, starting with an empty DAG and a topological ordering of~$S_\Fa$, and improve the score~$t$ times. This would be a polynomial time algorithm for instances where~$t$ is polynomial in~$n$ and~$S_\Fa$ is a DAG which would imply~$\text{P}=\NP$.

\section{Conclusion}

We initiated the study on parameterized ordering-based local search for the task of learning a Bayesian network structure. We studied four distances~$d$ and classified for which of these distances~\textsc{$d$-Local W-BNSL} is FPT and for which distances it is W[1]-hard.

There are several ways of extending our results that seem interesting topics for future
research. First, besides the experimental motivation mentioned in Section~\ref{Section: Experiments}, this work is purely theoretical. Recall that the running time bottleneck in these preliminary experiments was not the combinatorial explosion in the search radius~$r$ but rather the slow implementation of the insert operation. Thus, one important topic for future work is to investigate how well the theoretical algorithms proposed in this work perform in practice when combined with algorithm engineering tricks like data reduction rules, upper and lower bounds for the solution size, and restarts.

Besides experimental evaluations of the algorithms proposed in this work, it is also interesting to further study theoretical running time improvements. Recall that our subexponential-time algorithms for the inversions distance and the inversion-window distance are closely related to an algorithm for \textsc{FAST}~\cite{ALS09}. Feige~\cite{F09} and Karpinski and Schudy~\cite{KS10} improved this algorithm for \textsc{FAST} by proposing algorithms with running time~$2^{\Oh(\sqrt{k})}\cdot \poly(n)$ instead of~$2^{\Oh(\sqrt{k} \log(k))}\cdot \poly(n)$. Therefore, a natural open question is whether~\textsc{$d$-Local W-BNSL} for~$d \in \{ \text{Inv}, \text{InvWin}\}$ can be solved in~$2^{\Oh(\sqrt{r})} \cdot\poly(|I|)$~time.

Furthermore, the results in this work may be extended by considering further distances. One example for another distance could be a~$q$-swap distance, where a~$q$-swap for some fixed integer~$q$ is a swap of two vertices on the ordering which positions differ by at most~$q$. The~$q$-swap distance between two orderings is then the minimum number of~$q$-swaps needed to transform one ordering into the other. Observe that in case of~$q=1$, a~$q$-swap corresponds to an inversion. Observe that a~$q$-swap can be simulated by performing at most~$q$ inversions. Consequently, using Algorithm behind Theorem~\ref{Theorem: LB Inversions Algorithm} we can compute a solution in~$2^{\Oh(\sqrt{q \cdot r} \log(q \cdot r))} \cdot {|I|}^{\Oh(1)}$~time that is at least as good as the best solution that can be found by performing at most~$r$ $q$-swaps. One first question might be if there is a more efficient algorithm for~\textsc{$q${\rm -swap} Local W-BNSL}. Analogously to~$q$-swaps, one may define~$q$-inserts, where one can remove a vertex at position~$j$ and insert it at a new position~$i$ with~$j-q \leq i \leq j+q$. A hill-climbing strategy that is based on performing single~$q$-insert operatins has previously been studied by Alonso-Barba et al.~\cite{AOP11}.

In Section~\ref{Section: Limits}, we outlined the limits of ordering-based local search. In a nutshell, we discussed variants of BNSL that are NP-hard even on instances with an acyclic superstructure---in other words---instances, where the topological ordering of the solution is known. An idea to tackle such variants of BNSL might be to find parameters~$\ell$ for which they are FPT when the superstructure is acyclic. Then, it would be interesting to study parameterized local search parameterized by~$r+\ell$. Intuitively, there may be efficient algorithms that iterate over all~$r$-close orderings in~$n^{f(r)}$ time and then solve the instances restricted to each such ordering in~$g(\ell) \cdot \poly(|I|)$~time. Maybe it is also possible that this can be improved to an FPT-algorithm for parameterization by~$r+\ell$.

\bibliography{bnsl-local}
\bibliographystyle{plain}


\end{document}